\documentclass[a4paper,UKenglish,cleveref,thm-restate]{lipics-v2021}

\usepackage{bbm}

\usepackage[textsize=footnotesize,textwidth=1cm]{todonotes}

\renewcommand{\hat}[1]{\widehat{#1}}
\renewcommand{\bar}[1]{\overline{#1}}

\newcommand{\calC}{\mathcal{C}}
\newcommand{\calD}{\mathcal{D}}

\newcommand{\bbP}{\mathbb{P}}

\newcommand{\bbC}{\mathbb{C}}
\newcommand{\bbN}{\mathbb{N}}

\newcommand{\bbZ}{\mathbb{Z}}

\newcommand{\eps}{\epsilon}

\newcommand{\Ness}{\mathrm{N}_{\mathrm{ess}}}

\newcommand{\Ann}{\operatorname{Ann}}
\newcommand{\Span}{\operatorname{Span}}

\newcommand{\sat}{\operatorname{sat}}

\newcommand{\rank}{\mathrm{rank}}

\newcommand{\WR}{\mathsf{WR}}
\newcommand{\bwr}{\underline{\mathsf{WR}}}

\newcommand{\maxr}{\mathsf{maxR}}

\newcommand{\x}{\boldsymbol{x}}

\newcommand{\VBP}{\mathsf{VBP}}

\newcommand{\VP}{\mathsf{VP}}
\newcommand{\VWaring}{\mathsf{VWaring}}

\newcommand{\VNP}{\mathsf{VNP}}
\newcommand{\VF}{\mathsf{VF}}

\newcommand{\abpw}{\textup{\textsf{abpw}}}

\newcommand{\tr}{\operatorname{Tr}}

\newcommand{\CR}{\mathsf{CR}}
\newcommand{\bCR}{\underline{\mathsf{CR}}}
\newcommand{\SR}{\mathsf{SR}}

  \pdfoutput=1 
  \hideLIPIcs  
  
  \title{Fixed-parameter debordering of Waring rank} 
  
  \author{Pranjal Dutta}{School of Computing, National University of Singapore (NUS), Singapore \and \url{https://sites.google.com/view/pduttashomepage}}{pranjal@nus.edu.sg}{https://orcid.org/0000-0001-9137-9025}{Funded under the project ``Foundation of Lattice-based Cryptography'', by NUS-NCS Joint Laboratory for Cyber Security.}
  \author{Fulvio Gesmundo}{Institut de Math{\'e}matiques de Toulouse, Universit{\'e} Paul Sabatier, Toulouse, France \and \url{https://fulges.github.io/}}{fgesmund@math.univ-toulouse.fr}{https://orcid.org/0000-0001-6402-021X}{}
  \author{Christian Ikenmeyer}{University of Warwick, Warwick, UK  \and \url{https://www.dcs.warwick.ac.uk/~u2270030/}}{christian.ikenmeyer@warwick.ac.uk}{https://orcid.org/0000-0003-4654-177X}{Supported by EPSRC grant EP/W014882/1.}
  \author{Gorav Jindal}{Max Planck Institute for Software Systems, Saarbr{\"u}cken, Germany \and \url{https://goravjindal.github.io/}}{gjindal@mpi-sws.org}{https://orcid.org/0000-0002-9749-5032}{}
  \author{Vladimir Lysikov}{Ruhr-Universit\"at Bochum, Bochum, Germany \and \url{https://qi.rub.de/lysikov}}{Vladimir.Lysikov@ruhr-uni-bochum.de}{https://orcid.org/0000-0002-7816-6524}{Part of the work was done while V.L. was affiliated with the QMATH Centre, University of Copenhagen. V.L. acknowledges financial support from VILLUM FONDEN via the QMATH Centre of Excellence (Grant No. 10059) and the European Union (ERC Grant Agreements 818761 and 101040907).
  Views and opinions expressed are however those of the author(s) only and do not necessarily reflect those of the European Union or the European Research Council Executive Agency. Neither the European Union nor the granting authority can be held responsible for them.}

  \authorrunning{P. Dutta, F. Gesmundo, C. Ikenmeyer, G. Jindal, and V. Lysikov} 

  \Copyright{Pranjal Dutta, Fulvio Gesmundo, Christian Ikenmeyer, Gorav Jindal, and Vladimir Lysikov}


  \category{} 
   
  \relatedversion{} 
   
  \acknowledgements{}
  
  \nolinenumbers 
  
  \hypersetup{
      pdfsubject={}
  }

\begin{document}

\maketitle

\begin{abstract}
Border complexity measures are defined via limits (or topological closures), so that any function which can approximated arbitrarily closely by low complexity functions itself has low border complexity.
Debordering is the task of proving an upper bound on some non-border complexity measure in terms of a border complexity measure, thus getting rid of limits.

Debordering is at the heart of understanding the difference between Valiant's determinant vs permanent conjecture, and Mulmuley and Sohoni's variation which uses border determinantal complexity.
The debordering of matrix multiplication tensors by Bini played a pivotal role in the development of efficient matrix multiplication algorithms. Consequently, debordering finds applications in both establishing computational complexity lower bounds and facilitating algorithm design.
Currently, very few debordering results are known.

In this work, we study the question of debordering the border Waring rank of polynomials.
Waring and border Waring rank are very well studied measures in the context of invariant theory, algebraic geometry, and matrix multiplication algorithms.
For the first time, we obtain a Waring rank upper bound that is exponential in the border Waring rank and only {\em linear} in the degree. All previous known results were exponential in the degree.
For polynomials with constant border Waring rank, our results imply an upper bound on the Waring rank linear in degree, which previously was only known for polynomials with border Waring rank at most 5.
\end{abstract}

\newpage

\section{Introduction}
\label{sec:intro}

Given a homogeneous polynomial $f$ of degree $d$ over $\bbC$, its \emph{Waring rank} $\WR(f)$ is defined as the smallest number $r$ such that there exist homogeneous linear forms $\ell_1, \dots, \ell_r$ with
\[f = \sum_{i = 1}^r \ell_i^d.\]
Equivalently, $\WR(f)$ is the minimal top fanin of a homogeneous $\Sigma \wedge \Sigma$ circuit computing $f$.
In the case of quadratic forms (polynomials of degree $2$), this notion is equivalent to the rank of the symmetric matrix associated with a quadratic form;
hence Waring rank can be regarded as a generalization of the rank of a symmetric matrix.
Unlike the case of matrices, when $d \geq 3$, Waring rank is in general not lower semicontinuous \footnote{A function $f$ is lower semicontinuous at $a$ if $\underset{x \rightarrow a}{\lim \inf}~f(x) \ge f(a)$.}, that is, a limit of a family of polynomials with low Waring rank can have higher Waring rank.
The simplest example is given by the polynomial $x^{d - 1} y$, which has Waring rank $d$ (this is a classical result~\cite{Oldenburger}), but can be presented as a limit
\[
    x^{d - 1} y = \lim_{\eps \to 0} \frac{1}{d \eps} \left[(x + \eps y)^d - x^d\right]
\]
of a family of Waring rank $2$ polynomials (note that we work over $\bbC$, so this expression can be rearranged into a sum of two powers by moving constants inside the parentheses).
The \emph{border Waring rank} is a semicontinuous variation of Waring rank defined as follows:
the border Waring rank of~$f$, denoted $\bwr(f)$, is the smallest $r$ such that $f$ can be written as a limit of a sequence of polynomials of Waring rank at most $r$.
We have $\bwr(x^{d - 1} y) = 2$.

Waring rank was studied already in the eighteenth century \cite{Cay:TheoryLinTransformations,Sylv:PrinciplesCalculusForms,Cleb:TheorieFlachen} in the context of invariant theory, with the aim to determine normal forms for homogeneous polynomials. We mention the famous Sylvester Pentahedral Theorem, stating that a generic cubic form in four variables can be written uniquely as a sum of five cubes. At the beginning of the twentieth century, the early work on secant varieties in classical algebraic geometry \cite{Palatini:SuperficieAlg,Terr:seganti} implicitly commenced the study of border Waring rank. The notion of border rank for tensors was introduced in \cite{BCRL79} to construct faster-than-Strassen matrix multiplication algorithms. In \cite{Bin80}, Bini proved that tensor border rank and tensor rank define the same matrix multiplication exponent. Today this theory is deeply related to the study of Gorenstein algebras \cite{Iarrobino-Kanev,BuczBucz:SecantVarsHighDegVeroneseReembeddingsCataMatAndGorSchemes}, the Hilbert scheme of points \cite{Jeli:Pathologies}, and deformation theory \cite{BB-apolarity,JelMan:LimitsSaturatedIdeals}.

In the context of algebraic complexity theory, Waring rank defines a model of computation known as the homogeneous diagonal depth 3 circuits or homogeneous $\Sigma \wedge \Sigma$ circuits, see e.g.~\cite{sax08}.
This is a very weak computational model (determinants have provably exponential Waring rank~\cite{Gurvits-permanent}).
Nevertheless, it is important as one of the simplest nontrivial computational models and has many unresolved open problems associated with it.
The Waring rank of a generic homogeneous polynomial of degree $d \geq 3$ in $n$ variables is $\lceil \frac{1}{n} \cdot \binom{n+d-1}{d} \rceil = \Omega(\frac{n^{d - 1}}{d!})$
with finite number of exceptional values of $(n, d)$~\cite{AlHir:Poly_interpolation_in_several_variables},
but the best lower bounds obtained are of order $2 n^{\lfloor d/2 \rfloor}$ (from tensor rank lower bounds in~\cite{DBLP:conf/coco/AlexeevFT11}).
For a large class of lower bound methods, so called rank methods, there are barrier results showing that cannot give bounds significantly larger than $n^{\lfloor d/2 \rfloor}$~\cite{EGOW18,DBLP:conf/focs/GargMOW19,Galazka:VectorBundlesGiveEqnsForCactus}.
Waring rank can be useful when the degree of the polynomials considered is constant.
For example, the results of \cite{CHIL18} guarantee that the matrix multiplication exponent is controlled by the Waring rank or border Waring rank of the polynomial $\tr(X^3)$ with $X \in \bbC^{n \times n}$,
which is a symmetrized version of the matrix multiplication tensor.

The lack of semicontinuity is a common phenomenon in algebraic complexity not specific to Waring rank.
Most complexity measures defined in terms of discrete structures (such as circuits or formulas) or in terms of decompositions (such as Waring rank or tensor rank) are not lower semicontinuous.
To any algebraic complexity measure one can define the corresponding \emph{border complexity measure} in the same way as border Waring rank arises from Waring rank: the border complexity of $f$ is the smallest number $s$ such that $f$ can be approximated arbitrarily closely by polynomials of complexity at most $s$.
Border tensor rank appears in the study of the computational complexity of matrix multiplication~\cite{BCRL79,Bin80},
border complexity for algebraic circuits was first discussed in~\cite{MS01} and~\cite{Bur04}.

Replacing a complexity measure by its border measure in a complexity class, we obtain the \emph{closure} of this class. For example, $\overline{\VP}$ is the class of all polynomial sequences with polynomially bounded degree and border circuit size, and $\overline{\VF}$ is defined analogously using formula size.
Formally, the closure $\bar{\calC}$ of a complexity class $\calC$ consists of all polynomial sequences $(f_n)_{n \in \bbN}$
such that there exists a bivariate sequence $(g_{n, m})_{n, m \in \bbN}$ with the property that $(g_{n, m})_{n \in \bbN}$ lies in $\calC$ for every fixed $m$, and $f_n = \lim_{m \to \infty} g_{n, m}$.
The operation of going to the closure is indeed a closure operator in the sense of topology, see \cite{IS22}.

The relationship between border and non-border complexity is far from straightforward.
In some contexts taking a limit can be a very strong operation, which sometimes turns non-universal computational models into universal ones. For example, there are polynomials which cannot be computed by width~$2$ algebraic branching programs~\cite{DBLP:journals/cc/AllenderW16}, but the corresponding border measure is related to border formula size~\cite{BIZ18}, so every polynomial is a limit of width~$2$ ABPs. Kumar~\cite{kum20} gives an even easier example: every polynomial can be presented as a limit of a sum of $2$~products of affine linear forms.
On the other hand, there are examples of complexity measures which are lower semicontinuous, so that there is no difference between border and non-border complexity measures.
A simplest example is the number of monomials in a polynomial (equivalently, top fanin of a $\Sigma\Pi$ circuit).
Other examples are noncommutative ABP width (implicit in~\cite{nis91}) and read-once ABP width~\cite{duttademystifying}.

Semicontinuous complexity measures and closed complexity classes are easier to work with using geometric methods.
Because of this, the geometric complexity theory program~\cite{MS01} proposes to study conjectures $\VNP\not\subseteq\overline{\VBP}$ and $\VNP \not\subseteq \overline{\VP}$
instead of Valiant's conjectures $\VNP \neq \VBP$ and $\VNP \neq \VP$.
The $\VNP\not\subseteq\overline{\VP}$ conjecture was also proposed in~\cite{Bur04}.
These border variants of Valiant's conjecture are now usually referred to as the Mulmuley--Sohoni conjectures.
Mulmuley--Sohoni conjectures are stronger that Valiant's conjectures, but it is not clear how much stronger,
as most questions about the relations between complexity classes and their closures are wide open.
It is unknown even whether or not $\overline{\VF} \subseteq \VNP$.
Theorems of the form $\overline{\calC}\subseteq \calD$ for algebraic complexity classes $\calC$ and $\calD$ are called \emph{debordering} results.
These kind of results can also be proven directly on the complexity measures, by giving an upper bound on a non-border complexity in terms of border complexity.
For example, $\abpw(f) \leq \bwr(f)$, where $\abpw(f)$ is the algebraic branching program width of $f$.
This is proven using semicontinuity of noncommutative ABP width, see~\cite[Thm~4.2]{blaser2020complexity} and \cite{Forbes16}.
In terms of complexity classes, this means $\overline{\VWaring} \subseteq \VBP$, where $\VWaring$ is the class of $p$-families that have polynomially bounded Waring rank.

Forbes~\cite{wact-open-problems} conjectures that $\overline{\VWaring} = \VWaring$.
Since this puts $\overline{\VWaring}$ in $\VF$, a proof of this conjecture will also improve the results of Dutta, Dwivedi and Saxena~\cite{duttademystifying} from $\overline{\Sigma^{[r]}\Pi\Sigma} \subset \VBP$ to $\overline{\Sigma^{[r]}\Pi\Sigma} \subset \VF$.
Ballico and Bernardi~\cite{ballicobernardi2012} propose an even stronger conjecture stating that $\WR(f) \leq (\bwr(f) - 1) \cdot \deg f$.
This was proven by case analysis for small values of border Waring rank: for $\bwr(f) \leq 3$ in~\cite{landsberg2010}, for $\bwr(f) = 4$ in~\cite{ballicobernardi13}, and for $\bwr(f) = 5$ and $\deg f \geq 9$ in~\cite{ballico2019ranks}.

\subsection*{Main result}

We prove the following improved debordering theorem for border Waring rank.
\begin{theorem}[Fixed-parameter debordering]
\label{thm:intro:WR-deborder}
Let $f$ be a homogeneous polynomial with $\deg f = d$ and $\bwr(f) = r$. Then $\WR(f) \le 4^r \cdot d.$
\end{theorem}

Note that the example of the polynomial $x^{d - 1}y$ with $\bwr(x^{d - 1}y) = 2$ and $\WR(x^{d - 1}y) = d$ shows that any debordering bound must necessarily depend on both border Waring rank $r$ and the degree $d$.
We call our result a \emph{fixed-parameter debordering} because the bound is polynomial (in this case even linear) in $d$, but exponential in the complexity parameter $r$.
In the case of a fixed border Waring rank this gives a bound linear in the degree.
This was previously known only for $\bwr(f) \leq 5$.
Even for $r = O(\log d)$ we obtain an upper bound polynomial in $d$.

To the best of our knowledge, this is the first fixed-parameter debordering result.
Previous methods applied to border Waring rank only allow upper bounds of the order $d^r$ or~$r^d$.
To get $\WR(f) \leq O(d^r)$, note that a polynomial with border Waring rank $r$ can be transformed into a polynomial in only $r$ variables using a linear change of variables (see~\cref{thm:essvar}), and then take the maximal possible Waring rank of an $r$-variate polynomial of degree $d$ as an upper bound.
Alternatively, an upper bound $\WR(f) \leq 2^{d - 1} r^d$ can be obtained by using the previously mentioned debordering into an ABP ($\abpw(f) \leq \bwr(f)$) and writing the ABP as a sum of at most $r^d$ products, one for each path.
Other known debordering techniques, such as the interpolation technique using the bound on the degree of $\eps$ in the approximation from the work of Lehmkuhl and Lickteig~\cite{lehmkuhl1989order} (which is exponential in the degree of the polynomial), or the \textsf{DiDIL} technique from~\cite{duttademystifying} can be applied in the border Waring rank setting, but do not improve over the simpler results discussed above.

\subsection*{Proof ideas}

The main ideas for the proof come from \emph{apolarity theory} and the study of $0$-dimensional schemes in projective space (see~\cref{sec:apolarity}), but we managed to simplify the proof so that it is elementary and {\em does not use} the language of algebraic geometry and is based on partial derivative techniques (see~\cref{sec:borderrank}). 

To prove the debordering, we transform a border Waring rank decomposition for $f$ into a \emph{generalized additive decomposition}~\cite{Iarrobino_1995,BBM-ranks,BerOneTau:SchemesEvincedGAD} of the form
$
f = \sum_{k = 1}^m \ell_k^{d - r_k + 1} g_k,
$
where $\ell_k$~are linear forms, and $g_k$ are homogeneous polynomials of degrees $r_k - 1$.
We then obtain an upper bound on the Waring rank, by first decomposing each $g_k$ with respect to a basis consisting of powers of linear forms, and then using the classical fact (see also~\cite{BBT-monomials}) that $\WR(\ell_1^{a} \ell_2^b) \le \max(a+1,b+1)$. 

To construct a generalized additive decomposition, we divide the summands of a border rank decomposition into several parts such that cancellations happen \emph{only} between summands belonging to the same part; see~\cref{lem:borderrank-gad}. The key insight is that if the degree of polynomials involved is high enough, namely when $\deg f \geq \bwr(f) - 1$, then all parts of the decomposition are ``local'' in the sense that the lowest order term in each summand is a multiple of the same linear form. Each local part gives one term of the form $\ell^{d - r + 1} g$, where $r$ is the number of rank one summands in the part and $\ell$ is the common lowest order linear form; see~\cref{lem:local-limit}. 

For example, consider the family of polynomials $f_d = x_0^{d - 1} y_0 + x_1^{d - 1} y_1 + 2 (x_0 + x_1)^{d - 1} y_2$, adapted from~\cite{BB-wild}.
If $d > 3$, then the border Waring rank of $f$ is at most $6$, as evidenced by the decomposition
\begin{equation}\label{eq:bwr-example-highdegree}
f_d = \lim_{\eps \to 0} \frac{1}{d \eps}\left[(x_0 + \eps y_0)^d - x_0^d + (x_1 + \eps y_1)^d - x_1^d + 2 (x_0 + x_1 + \eps y_2)^d - 2 (x_0 + x_1)^d \right],
\end{equation}
and a matching lower bound is obtained by considering the dimension of the space of second order partial derivatives.
The summands of the decomposition~\eqref{eq:bwr-example-highdegree} can be divided into three pairs. The lowest order term of the first pair is $x_0^d$, the one of the second pair is $x_1^d$ and the one of the third pair is $(x_0+x_1)^d$.
For each pair, the sum of the two powers individually converges to a limit as $\eps \to 0$; these three limits are, respectively, $x_0^{d-1} y_0$, $x_1^{d - 1} y_1$, and $2 (x_0 + x_1)^{d - 1} y_2$, which are the summands of a generalized additive decomposition for $f_d$.

When $d = 3$, the polynomial $f_d$ is an example of a ``wild form''~\cite{BB-wild}.
It has border Waring rank $5$ given for example by the decomposition
\begin{multline}\label{eq:bwr-example-wild}
f_3 = \lim_{\eps \to 0} \frac{1}{9 \eps} \left[ 3 (x_0 + \eps y_0)^3 + 3 (x_1 + \eps y_1)^3 + \right. \\ 
    \left. 6 (x_0 + x_1 + \eps y_2)^3 - (x_0 + 2 x_1)^3 - (2x_0 + x_3)^3 \right].
\end{multline}
Unlike the previous decomposition, this one cannot be divided into parts that have limits individually,
and is not local --- all summands have different lowest order terms.
This is only possible if the degree is low.

The condition on the degree is related to algebro-geometric questions about regularity of $0$-dimensional schemes~\cite[Thm.~1.69]{Iarrobino-Kanev}, but for the schemes arising from border rank decompositions, this is ultimately a consequence of the fact that $r$ distinct linear forms have linearly independent $d$-th powers when $d \geq r - 1$. 

\section{Debordering border Waring rank}
\label{sec:deborderingWR}

The goal of this section is to prove~\cref{thm:intro:WR-deborder}. Given a homogeneous degree $d$ polynomial $f$, we provide upper bounds for $\WR(f)$ in terms of $\bwr(f)$ and $d$.

\subsection{Definitions}

In this section we introduce some notation and give a formal definition of Waring rank and border Waring rank.
We work over the field $\bbC$ of complex numbers. The space of homogeneous polynomials of degree $d$ in variables $\x = (x_1, \dots, x_n)$ is denoted by $\bbC[\x]_d$.
We write $f \simeq g$ for $f, g \in \bbC(\eps)[\x]$ if $\lim_{\eps \to 0} f = \lim_{\eps \to 0} g$ (in particular, both limits must exist).
Recall that the projective space $\bbP V$ is defined as the set of lines through the origin in $V$, that is, for each nonzero $v \in V$ we have a corresponding line $[v] \in \bbP V$, and $[v] = [w]$ if and only if $v = \alpha w$ for some $\alpha$.

\begin{definition}
A \emph{Waring rank decomposition} of a homogeneous polynomial $f \in \bbC[\x]_d$ is a decomposition of the form
\[
f = \sum_{k = 1}^r \ell_k^d
\]
for some linear forms $\ell_1, \dots, \ell_r \in \bbC[\x]_1$.
The minimal number of summands in a Waring rank decomposition is called the \emph{Waring rank} of $f$ and is denoted by $\WR(f)$.
\end{definition}
It is known that every homogeneous polynomial over $\bbC$ has finite Waring rank~\cite{Oldenburger}.

\begin{definition}
A \emph{border Waring rank decomposition} of a homogeneous polynomial $f \in \bbC[\x]_d$ is an expression of the form
\[
f = \lim_{\eps \to 0} \sum_{k = 1}^r \ell_k^d
\]
where $\ell_1, \dots, \ell_r \in \bbC(\eps)[\x]_1$, that is, $\ell_i$ are linear forms in $\x$ with coefficients rationally dependent on $\eps$.
The \emph{border Waring rank} $\bwr(f)$ is the minimal number of summands in a border Waring rank decomposition.
\end{definition}
Equivalently, the border Waring rank of $f \in \bbC[\x]_d$ can be defined as the minimal number $r$ such that $f$ lies in the closure of the set $W_{d, r} = \{ g \in \bbC[\x]_d \mid \WR(g) \leq r \}$ of all polynomials with Waring rank at most $r$.
The set $W_{d, r}$ is constructible, so its Zariski and Euclidean closures coincide, see e.g.~\cite[Anh.I.7.2 Folgerung]{Kra85}.
The equivalence to the definition given above was established by Alder~\cite{alder1984} (cited by~\cite[Ch.20]{bcs97}) for a similar notion of tensor rank, the proof remains essentially the same for Waring rank of polynomials.

\subsection{Orbit closure and essential variables}\label{sec:essvardebord}

The number of essential variables of a homogeneous polynomial $f$ is the minimum integer $m$ such that there is a linear change of coordinates after which $f$ can be written as a polynomial in $m$ variables.
Denote the number of essential variables of $f$ by $\Ness(f)$.
It is a classical fact, which already appears in \cite{Sylv:PrinciplesCalculusForms}, that the number of essential variables of $f$ equals the dimension of the linear span of its first order partial derivatives, or equivalently the rank of the first partial derivative map. In particular $\Ness(-)$ is a lower semicontinuous function. We refer to \cite{carlini2006reducing} and~\cite[Lemma~B.1]{kayal2007polynomial} for modern proofs of this result.

An immediate consequence of the semicontinuity of the number of essential variables is the following result.
\begin{lemma}
\label{thm:essvar}
For a homogeneous polynomial $f \in \bbC[\x]_d$ we have $\Ness(f) \leq \bwr(f)$.
\end{lemma}
\begin{proof}
We first prove $\Ness(f) \leq \WR(f)$.
Let $p$ be the dimension of the linear space spanned by the linear forms $\ell_k$ in the decomposition $f = \sum_{k = 1}^r \ell_k^d$. 
Without loss of generality the linear forms $\ell_1, \dots, \ell_p$ are linearly independent, and $\ell_{p + 1}, \dots, \ell_r$ are linear combinations of~$\ell_1, \dots, \ell_p$.
After applying a change of variables such that $y_k = \ell_k(\x)$ for $k = 1, \dots, p$ we see that $\Ness(f) \leq p \leq r$.

The inequality $\Ness(f) \leq \bwr(f)$ now follows from the semicontinuity of $\Ness$: if 
\[
f = \lim_{\eps \to 0} \sum_{k = 1}^r \ell_k^d\;,
\] with $\ell_k \in \bbC(\eps)[\x]$, then $\Ness(f) \leq \lim_{\eps \to 0} \Ness(\sum_{k = 1}^r \ell_k^d(\eps)) \leq r$.
\end{proof}

\subsection{Fixed-parameter debordering}\label{sec:borderrank}

The proof of \cref{thm:intro:WR-deborder} is based on generalized additive decompositions of polynomials, in the sense of \cite{Iarrobino_1995}. These decompositions were studied in algebraic geometry, usually in connection to $0$-dimensional schemes and the notion of cactus rank. We defer the discussion of connections to algebraic geometry to the next section. Here we provide elementary proofs of some statements on generalized additive decompositions based on partial derivatives techniques, without using the language of $0$-dimensional schemes. We bring from geometry a key insight: a border rank decomposition can be separated into \emph{local} parts if the degree of the polynomial is large enough.

To define formally what it means for a border rank decomposition to be local, note that a rational family of linear forms $\ell \in \bbC(\eps)[\x]_1$ always has a limit when viewed projectively.
Specifically, expanding $\ell(\eps)$ as a Laurent series $\ell(\eps) = \sum_{i = q}^{\infty} \eps^{i} \ell_i$ with $\ell_q \neq 0$, we have $\lim_{\eps \to 0} [\ell(\eps)] = \lim_{\eps \to 0} [\sum_{i = 0}^{\infty} \eps^{i} \ell_{q + i}] = [\ell_q]$.
A border Waring rank decomposition is called local if for all summands in the decomposition this limit is the same. More precisely, we give the following definition.

\begin{definition}
Let $f \in \bbC[\x]_d$ be a homogeneous polynomial.
A border Waring rank decomposition
\[
f = \lim_{\eps \to 0} \sum_{k = 1}^r \ell_k^{d}\;,
\]
with $\ell_k \in \bbC(\eps)[\x]_1$
is called a \emph{local border decomposition} if there exists a linear form $\ell \in \bbC[\x]_1$ such that $\lim_{\eps \to 0} [\ell_k(\eps)] = [\ell]$ for all $k \in \{1, \dots, r\}$.
We call the point $[\ell] \in \bbP\bbC[\x]_1$ the \emph{base} of the decomposition.
A local decomposition is called \emph{standard} if $\ell_1 = \eps^q \gamma \ell$ for some $q \in \bbZ$ and $\gamma \in \bbC$.
\end{definition}

\begin{lemma}
If $f$ has a local border decomposition, then it has a standard local border decomposition with the same base and the same number of summands.
\end{lemma}
\begin{proof}
After applying a linear change of variables, we may assume that the base of the local decomposition for $f$ is $[x_1]$.
This means
\[
f = \lim_{\eps \to 0} \sum_{k = 1}^r \ell_k^{d}
\]
with $\ell_k = \eps^{q_k} \cdot \gamma_k x_1 + \sum_{j = q_k + 1}^{\infty} \eps^{j} \ell_{k,j}$.

Write $\ell_1 = \eps^{q_1}\left(\sum_{i = 1}^n \alpha_i x_i \right)$ where $\alpha_i \in \bbC(\eps)$.
Let $\hat{x}_1 = \frac{\gamma_1}{\alpha_1} x_1 - \sum_{i = 2}^n  \frac{\alpha_n}{\alpha_1} x_i$.
Note that $\alpha_1 \simeq \gamma_1$ and $\alpha_i \simeq 0$ for $i > 1$, hence $\hat{x}_1 \simeq x_1$ and
\[
f \simeq f(\hat{x}_1, \dots, x_n) \simeq \ell_1(\hat{x}_1, x_2, \dots, x_n)^d + \sum_{k = 2}^r \ell_k(\hat{x}_1, x_2, \dots, x_n)^d = (\eps^{q_1} \gamma_1 x_1)^d + \sum_{k = 2}^{r} \hat{\ell}_k^d.
\]
where $\hat{\ell}_k(x_1, \dots, x_n) = \ell_k(\hat{x}_1, x_2, \dots, x_n)$.
This defines a new border rank decomposition of~$f$. Moreover, notice that $\lim_{\eps \to 0} [\hat{\ell}_k]= [x_1]$ for every $k$, so the new decomposition is again local with base $[x_1]$.
Since the first summand is $\eps^{q_1} \gamma_1 x_1$, this is the desired standard local border decomposition.
\end{proof}

\begin{lemma}\label{lem:local-limit}
Suppose $f \in \bbC[\x]_d$ has a local border decomposition with $r$ summands based at~$[\ell]$.
If $d \geq r - 1$, then $f = \ell^{d - r + 1} g$ for some homogeneous polynomial $g$ of degree $r - 1$.
\end{lemma}
\begin{proof}
After applying a linear change of variables we may assume $\ell = x_1$.
We prove the statement by induction on $r$ and the difference $d - (r - 1)$.

The cases $r = 1$ and $d = r - 1$ are trivial.

If $d > r - 1$, then by the previous Lemma there exists a standard local border decomposition
\[
f = \lim_{\eps \to 0} \sum_{k = 1}^r \ell_k(\eps)^d.
\]
Write $\ell_k = \sum_{i = 1}^n \alpha_{ki} x_i$ for some $\alpha_{ki} \in \bbC(\eps)$. Since the decomposition is standard, $\alpha_{1i} = 0$ for $i > 1$.
For the derivatives of $f$ we have the following border decompositions
\[
\frac{\partial f}{\partial x_1} = \lim_{\eps \to 0} \sum_{k = 1}^r d \cdot \alpha_{k 1}(\eps) \ell_k(\eps)^{d-1},
\]
and
\[
\frac{\partial f}{\partial x_i} = \lim_{\eps \to 0} \sum_{k = 2}^r d \cdot \alpha_{ki}(\eps) \ell_k(\eps)^{d-1}.
\]
for $i \neq 1$.
These decompositions involve the same linear forms $\ell_k$ with multiplicative coefficients, so they are local with the same base $[x_1]$.
By inductive hypothesis $\frac{\partial f}{\partial x_1} = x_1^{d - r} g_1$ and $\frac{\partial f}{\partial x_i} = x_1^{d - r + 1} g_i$ for some homogeneous polynomials $g_1, \dots, g_n$ of appropriate degrees.
To get an analogous expression for $f$, combine these expressions using Euler's formula for homogeneous polynomials as follows
\[
f = \frac{1}{d} \sum_{i = 1}^n x_i \frac{\partial f}{\partial x_i} = \frac{1}{d} \left( x_1 \cdot x_1^{d - r} g_1 + \sum_{i = 2}^n x_i x_1^{d - r + 1} g_i \right) = \frac{1}{d} x_1^{d - r + 1} \left(g_1 + \sum_{i = 2}^n x_i g_i\right)\;.
\qedhere\]
\end{proof}

We will now extend this result to non-local border Waring rank decompositions.
As long as the degree of the approximated polynomial is high enough, every border rank decomposition can be divided into local parts and transformed into a sum of terms of the form $\ell^{d - r + 1} g$.

\begin{definition}
A \emph{generalized additive decomposition} of $f$ is a decomposition of the form
\[
f = \sum_{k = 1}^m \ell_k^{d - r_k + 1} g_k\;,
\]
where $\ell_k$ are linear forms such that $\ell_i$ is not proportional to $\ell_j$ when $i \neq j$, and $g_k$ are homogeneous polynomials of degrees $\deg g_k = r_k - 1$.
\end{definition}

To show that there are no cancellations between different local parts, we need the following lemma, which in the case of $2$ variables goes back to Jordan~\cite[Lem.~1.35]{Iarrobino-Kanev}. This lemma can be seen as a generalization of a well-known fact that $m$ pairwise non-proportional linear forms $\ell_1, \dots, \ell_m$ have linearly independent powers $\ell_1^d, \dots, \ell_m^d$ for $d \geq m - 1$.

\begin{lemma}\label{lem:jordan-lemma}
Let $\ell_1, \dots, \ell_m \in \bbC[\x]_1$ be linear forms such that $\ell_i$ is not proportional to~$\ell_j$ when $i \neq j$.
Let $g_1, \dots, g_m$ be homogeneous polynomials of degrees $r_1 - 1, \dots, r_m - 1$ respectively.
If
\[
\sum_{k = 1}^m \ell_k^{d - r_k + 1} g_k \;=\; 0\;,
\]
and $d \geq \sum_{k = 1}^m r_i - 1$, then all $g_k$ are zero.
\end{lemma}
\begin{proof}
We first prove the statement for polynomials in $2$ variables $y_1, y_2$ by induction on the number of summands $m$; this part of the proof closely follows~\cite[Appx.III]{grace_young_2010}.

The case $m = 1$ with one summand is clear.
Consider the case $m > 2$.
We can assume $\ell_1 = y_1$ by applying a linear change of variables if required.
Note two simple facts about partial derivatives.
First, for a homogeneous polynomial $f \in \bbC[y_1, y_2]_d$ we have $\partial_2^r f = 0$ if and only if $f = y_1^{d - r + 1} g$ (here $\partial_2 := \frac{\partial}{\partial y_2}$).
Second, differentiating $r$ times a homogeneous polynomial of the form $\ell^{d - s + 1} g$, we obtain a polynomial of the form $\ell^{d - r - s + 1} h$.

Suppose
\[
y_1^{d - r_1 + 1} g_1 + \sum_{k = 2}^m \ell_k^{d - r_k + 1} g_k = 0.
\]
Differentiating $r_1$ times with respect to $y_2$, we obtain
\[
\sum_{k = 2}^m \ell_k^{d - r_1 - r_k + 1} h_k = 0,
\]
where $\ell_k^{d - r_1 - r_k + 1} h_k = \partial_2^{r_1} (\ell_k^{d - r_k + 1} g_k)$.
The degree condition $d - r_1 \geq \sum_{k = 2}^m r_k - 1$ holds for this new expression.
Therefore, by induction hypothesis we have $h_k = 0$ and thus $\partial_2^{r_1} (\ell_k^{d - r_k + 1} g_k) = 0$.
It follows that $\ell_k^{d - r_k + 1} g_k = y_1^{d - r_1 + 1} \hat{g}_k$ for some homogeneous polynomial~$\hat{g}_k$.
This implies that $y_1^{d - r_1 + 1}$ divides $g_k$, which is impossible since $d - r_1 + 1 \geq \sum_{k = 2}^m r_k \geq r_k > \deg g_k$.

Consider now the general case where the number of variables $n \geq 2$ (the case $n = 1$ is trivial).
Suppose $\sum_{k = 1}^m \ell_k^{d - r_k + 1} g_k = 0$.
The set of linear maps $A \colon (y_1, y_2) \mapsto (x_1, \dots, x_n)$ such that $\ell_i \circ A$ and $\ell_j \circ A$ are not proportional to each other is a nonempty Zariski open set given by the condition $\rank(\ell_i \circ A, \ell_j \circ A) > 1$.
Hence for a nonempty Zariski open (and therefore dense) set of linear maps $A$ the linear forms $\ell_k \circ A$ are pairwise non-proportional.
From the binary case above we have $g_k \circ A = 0$ if $A$ lies in this open set.
By continuity this implies $g_k \circ A = 0$ for all $A$.
Since every point lies in the image of some linear map $A$ we have $g_k = 0$.
\end{proof}

\begin{restatable}{lemma}{mainlemma}\label{lem:borderrank-gad}
Let $f \in \bbC[\x]_d$ be such that $\bwr(f) = r$.
If $d \geq r - 1$, then there exists a partition $r = r_1 + \dots + r_m$ such that $f$ has a generalized additive decomposition
\[
f = \sum_{k = 1}^m \ell_k^{d - r_k + 1} g_k,
\]
and moreover $\bwr(\ell_k^{d - r_k + 1} g_k) \leq r_k$.
\end{restatable}
\begin{proof}
Consider a border Waring rank decomposition
\[
f = \lim_{\eps \to 0} \sum_{k = 1}^r \ell_k^{d}
\]
Divide the summands between several local decompositions as follows.
Define an equivalence relation $\sim$ on the set of indices $\{1,2,\dots,r\}$ as $i \sim j \Leftrightarrow \lim_{\eps \to 0} [\ell_i] = \lim_{\eps \to 0} [\ell_j]$ and let $I_1, \dots, I_m$ be the equivalence classes with respect to this relation.
Further, let $r_k = |I_k|$ and let $[L_k] = \lim_{\eps \to 0} [\ell_i]$ for $i \in I_k$.

Consider the sum of all summands with indices in $I_k$.
Let $q_k$ be the power of $\eps$ in the lowest order term, that is,
\[
\sum_{i \in I_k} \ell_i^d \;=\; \eps^{q_k} f_k + \sum_{j = q_k + 1}^{\infty} \eps^j f_{k, j}\;,
\] with $f_k \in \bbC[\mathbf{x}]_d$ nonzero.
This expression can be transformed into a local border decomposition
\[
f_k = \lim_{\eps \to 0} \sum_{i \in I_k} \left(\frac{\ell_i(\eps^d)}{\eps^{q_k}}\right)^d.
\]
based at $[L_k]$.
By~\Cref{lem:local-limit} we have $f_k = L_k^{d - r_k + 1} g_k$ for some homogeneous polynomial $g_k$ of degree $r_k - 1$. The decomposition also gives $\bwr(f_k) \leq r_k$.

Note that $q_k \leq 0$ since otherwise the summands $\ell_i$ with $i \in I_k$ can be removed from the original border rank decomposition of $f$ without changing the limit.
Let $q = \min\{q_1, \dots, q_m\}$.
Note that if $q < 0$, then, comparing the terms before $\eps^q$ in the left and right hand sides of the equality
\[
f + O(\eps) = \sum_{k = 1}^m \sum_{i \in I_k} \ell_i^d
\]
we get 
\[
0 = \sum_{k\colon q_k = q} f_k = \sum_{k\colon q_k = q} L_k^{d - r_k + 1} g_k.
\]
From~\Cref{lem:jordan-lemma} we obtain $g_k = 0$ and $f_k = 0$, in contradiction with the definition of $f_k$.

We conclude that $q = 0$ and 
\[
f = \sum_{k = 1}^m f_k = \sum_{k = 1}^m L_k^{d - r_k + 1} g_k,
\]
obtaining the required generalized additive decomposition.
\end{proof}

We will now take a brief detour to define a function $M(r)$ which we use to upper bound the Waring rank of generalized additive decomposition.

\begin{definition}
Let $\maxr(n, d)$ denote the \emph{maximum Waring rank} of a degree $d$ homogeneous polynomial in $n$ variables, that is $\maxr(n, d) = \max\{\WR(f)\mid f \in \bbC[x_1,\dots, x_n]_d\}$.
Define the \emph{partition-maxrank function} as
\[
M(r) = \max_{r_1 + \dots + r_m = r} \sum_{k = 1}^m \maxr(r_k, r_k - 1).
\]
\end{definition}

Since every homogeneous polynomial has finite Waring rank, the space $\bbC[x_1, \dots, x_n]_d$ is spanned by powers of linear forms. This implies a trivial upper bound on the maximum Waring rank: $\maxr(n, d) \leq \dim \bbC[x_1, \dots, x_n]_d = \binom{n + d - 1}{d}$.
Improved upper bounds were proven in~\cite{BlekTeit:MaximumTypicalGenericRanks,Jelisiejew-maxrank}.

\begin{proposition}
$\maxr(n, d_1) \leq \maxr(n, d_2)$ when $d_1 \leq d_2$.
\end{proposition}
\begin{proof}
Every form $f$ of degree $d_1$ can be represented as a partial derivative of some form $g$ of degree $d_2$. By differentiating a Waring rank decomposition of $g$ we obtain a Waring rank decomposition of $f$, thus $\WR(f) \leq \WR(g) \leq \maxr(n, d_2)$. Since $f$ is arbitrary, $\maxr(n, d_1) \leq \maxr(n, d_2)$.
\end{proof}

We are now ready to prove a debordering theorem for Waring rank.

\begin{theorem}\label{thm:waring-to-maxwaring}
Let $f \in \bbC[\x]_d$ be such that $\bwr(f) = r$. Then 
 \[
  \WR(f) \leq M(r) \cdot d.
 \]
\end{theorem}
\begin{proof}
We consider two cases depending on relation of degree $d$ and border Waring rank $r$.

\emph{Case} $d < r - 1$. Since $\bwr(f) = r$, the number of essential variables of $f$ is at most $r$.
Taking the maximum Waring rank as an upper bound, we obtain 
\[
\WR(f) \leq \maxr(r, d) < \maxr(r, r - 1) \leq M(r) \leq M(r) \cdot d.
\]

\emph{Case} $d \geq r - 1$. By~\Cref{lem:borderrank-gad} $f$ has a generalized additive decomposition
\[
f = \sum_{k = 1}^m \ell_k^{d - r_k + 1} g_k
\]
with $r_1 + \dots + r_m = r$, $\deg g_k = r_k - 1$ and $\bwr(\ell_k^{d - r_k + 1} g_k) \leq r_k$.
Since $\bwr(\ell_k^{d - r_k + 1} g_k) \leq r_k$, the number of essential variables $\Ness(g_k) \leq r_k$.
If $r_k = 1$, then 
\[
\WR(\ell_k^{d - r_k + 1} g_k) = \WR(\ell_k^d) = 1 \leq d.
\]
If $r_k \geq 2$, then we upper bound $\WR(g_k)$ by $\maxr(\Ness(g_k), \deg g_k) = \maxr(r_k, r_k - 1)$.
Taking a Waring rank decomposition $g_k = \sum_{i = 1}^{\WR(g_k)} L_i^{r_k - 1}$ and multiplying it by $\ell_k^{d - r_k + 1}$, we obtain a decomposition 
\[
\ell_k^{d - r_k + 1} g_k = \sum_{i = 1}^{\WR(g_k)} \ell_k^{d - r_k + 1} \cdot L_i^{r_k - 1}.
\]
It is known that $\WR(y_1^a y_2^b) = \max\{a, b\} + 1$ (this is a classical fact known at least to Oldenburger~\cite{Oldenburger}, see also~\cite{BBT-monomials})\footnote{it is easy to see that for $a \geq b$ the monomial $y_1^a y_2^b$ is proportional to $\sum_{k = 0}^{a} \zeta^k (\zeta^k y_1 + y_2)^{a + b}$ where $\zeta$ is a primitive root of unity of order $a + 1$.}.
It follows that
\[
\WR(\ell_k^{d - r_k + 1} L_i^{r_k - 1}) \leq \WR(y_1^{d - r_k + 1} y_2^{r_k - 1}) = \max\{d - r_k + 2, r_k\} \leq d.
\]
Hence we have $\WR(\ell_k^{d - r_k + 1} g_k) \leq d \cdot \WR(g_k) \leq d \cdot \maxr(r_k - 1, r_k)$. 

Combining all parts of the decomposition together, we get 
\[
\WR(f) \leq d \sum_{k = 1}^m \maxr(r-k - 1, r_k) \leq M(r) \cdot d.
\qedhere\]
\end{proof}

A more explicit upper bound is provided by the following immediate corollary.
\begin{theorem}
\label{thm:WR-deborder}
Let $f \in \bbC[x_1, \dots, x_n]_d$
and let $\bwr(f) = r$. Then 
 \[
  \WR(f) \leq \binom{2r - 2}{r - 1} \cdot d.
 \]
\end{theorem}
\begin{proof}
The space of homogeneous polynomials of degree $r - 1$ in $r$ variables has dimension $\binom{2r - 2}{r - 1}$ and is spanned by powers of linear forms.
Therefore, $\maxr(r - 1, r) \leq \binom{2r - 2}{r - 1}$.
Note that if $r = p + q$ with $p, q \neq 0$, then the space $\bbC[x_1, \dots, x_r]_{r - 1}$ contains a direct sum of $x_1^{q} \cdot\bbC[x_1, \dots, x_p]_{p - 1}$ and $x_1^{p + 1} \cdot \bbC[x_{p+1}, \dots, x_r]_{q - 1}$.
Taking the dimensions of these spaces, we obtain
$\binom{2r - 2}{r - 1} \geq \binom{2p - 2}{p - 1} + \binom{2q - 2}{q - 1}$.
It follows that $M(r) \leq \binom{2r - 2}{r - 1}$.
\end{proof}

Using the Blekherman--Teitler bound on the maximum rank~\cite{BlekTeit:MaximumTypicalGenericRanks}, we can get a slightly better bound. The proof is essentially the same as for the previous theorem.

\begin{corollary}
Let $f \in \bbC[x_1, \dots, x_n]_d$
and let $\bwr(f) = r$. Then 
 \[
  \WR(f) \leq 2 \left\lceil \frac{1}{r} \binom{2r - 2}{r - 1} \right\rceil \cdot d.
 \]
\end{corollary}

\subsection{Scheme-theoretic proof}

In this section we give a proof of~\cref{lem:borderrank-gad} based on the theory of $0$-dimensional schemes and apolarity.
This short section assumes familiarity with these topics, we review them in more details in~\cref{sec:apolarity}.

\mainlemma*
\begin{proof}[Alternative proof]
Denote by $V$ the space of linear forms $\bbC[\x]_1$.

Since $d \geq r - 1$, the border Waring rank of $f$ is equal to its \emph{smoothable rank} $\SR(f)$~\cite{BB-wild}, that is, there exists a $0$-dimensional scheme $Z \subset \bbP V$ of degree $r$ which is smoothable (obtained as a flat limit of the family of $r$-point subsets of $\bbP V$) and $f$ is apolar to $Z$. 
Let $I$ be the ideal of $Z$ and let $I = I^{(1)} \cap \dots \cap I^{(m)}$ be the primary decomposition of this ideal.
The primary ideals $I^{(j)}$ correspond to irreducible components $Z_j$ of the scheme $Z$.

Since $f$ is apolar to $I$, we have $f \in I_d^{\perp} = (I^{(1)}_d)^{\perp} + \dots + (I^{(m)}_d)^{\perp}$.
In particular, there exist $f_j \in (I^{(j)}_d)^{\perp}$ such that $f = f_1 + \cdots + f_m$.
Let~$r_j$ be the degree of $Z_j$. By the definition of degree, $r = r_1 + \cdots + r_m$.
If $Z_j$ is supported at the point $[\ell_j] \in \bbP V$, then for the ideal $I^{(j)}$ we have $(\ell_j^{\perp})^{r_j} \subset I^{(j)} \subset \ell_j^{\perp}$ and $(I^{(j)}_d)^{\perp} \subset \ell_j^{d - r_j + 1} \cdot \bbC[\x]_{r_j - 1}$.
Therefore the polynomials $f_j$ have the form $\ell_j^{d - r_j + 1} g_j$ for some $g_j$ of degree $\deg g_j = r_j - 1$.

Additionally, all irreducible components of a smoothable scheme $Z$ are smoothable~\cite[Thm. 1.1]{MR3724212}, and since $f_j$ is apolar to $Z_j$, we have $\bwr(f_j) \leq \SR(f_j) \leq r_j$.
\end{proof}

\bibliography{stacs2024.bib}

\begin{thebibliography}{10}

\bibitem{alder1984}
Alexander Alder.
\newblock {\em Grenzrang und Grenzkomplexit{\"a}t aus algebraischer und
  topologischer Sicht}.
\newblock PhD thesis, Universität Zürich, 1984.

\bibitem{AlHir:Poly_interpolation_in_several_variables}
J.~Alexander and A.~Hirschowitz.
\newblock {Polynomial interpolation in several variables}.
\newblock {\em J. Alg. Geom.}, 4(2):201--222, 1995.

\bibitem{DBLP:conf/coco/AlexeevFT11}
Boris Alexeev, Michael~A. Forbes, and Jacob Tsimerman.
\newblock Tensor rank: Some lower and upper bounds.
\newblock In {\em {CCC} 2011}, pages 283--291. {IEEE} Computer Society, 2011.
\newblock \href {https://doi.org/10.1109/CCC.2011.28}
  {\path{doi:10.1109/CCC.2011.28}}.

\bibitem{DBLP:journals/cc/AllenderW16}
Eric Allender and Fengming Wang.
\newblock On the power of algebraic branching programs of width two.
\newblock {\em Comput. Complex.}, 25(1):217--253, 2016.
\newblock \href {https://doi.org/10.1007/s00037-015-0114-7}
  {\path{doi:10.1007/s00037-015-0114-7}}.

\bibitem{ballico2019ranks}
Edoardo Ballico.
\newblock On the ranks of homogeneous polynomials of degree at least 9 and
  border rank 5.
\newblock {\em Note di Matematica}, 38(2):55--92, 2019.
\newblock \href {https://doi.org/10.1285/i15900932v38n2p55}
  {\path{doi:10.1285/i15900932v38n2p55}}.

\bibitem{ballicobernardi13}
Edoardo Ballico and Alessandra Bernardi.
\newblock Stratification of the fourth secant variety of {V}eronese varieties
  via the symmetric rank.
\newblock {\em Adv. Pure Appl. Math.}, 4(2):215--250, 2013.
\newblock \href {https://doi.org/10.1515/apam-2013-0015}
  {\path{doi:10.1515/apam-2013-0015}}.

\bibitem{ballicobernardi2012}
Edoardo Ballico and Alessandra Bernardi.
\newblock Curvilinear schemes and maximum rank of forms.
\newblock {\em Le Matematiche}, 72(1):137--144, 2017.
\newblock \href {https://doi.org/10.4418/2017.72.1.10}
  {\path{doi:10.4418/2017.72.1.10}}.

\bibitem{BBM-ranks}
Alessandra Bernardi, J\'{e}r\^{o}me Brachat, and Bernard Mourrain.
\newblock A comparison of different notions of ranks of symmetric tensors.
\newblock {\em Linear Algebra Appl.}, 460:205--230, 2014.
\newblock \href {https://doi.org/10.1016/j.laa.2014.07.036}
  {\path{doi:10.1016/j.laa.2014.07.036}}.

\bibitem{Bernardi2018}
Alessandra Bernardi, Joachim Jelisiejew, Pedro Macias~Marques, and Kristian
  Ranestad.
\newblock On polynomials with given {H}ilbert function and applications.
\newblock {\em Collect. Math.}, 69(1):39--64, 2018.
\newblock \href {https://doi.org/10.1007/s13348-016-0190-2}
  {\path{doi:10.1007/s13348-016-0190-2}}.

\bibitem{BerOneTau:SchemesEvincedGAD}
Alessandra Bernardi, Alessandro Oneto, and Daniele Taufer.
\newblock {On schemes evinced by generalized additive decompositions and their
  regularity}, 2023.
\newblock \href {http://arxiv.org/abs/2309.12961} {\path{arXiv:2309.12961}}.

\bibitem{Bin80}
Dario Bini.
\newblock Relations between exact and approximate bilinear algorithms.
  {A}pplications.
\newblock {\em Calcolo}, 17(1):87--97, 1980.
\newblock \href {https://doi.org/10.1007/BF02575865}
  {\path{doi:10.1007/BF02575865}}.

\bibitem{BCRL79}
Dario Bini, Milvio Capovani, Francesco Romani, and Grazia Lotti.
\newblock {$O(n^{2.7799})$} complexity for $n \times n$ approximate matrix
  multiplication.
\newblock {\em Inf. Process. Lett.}, 8(5):234--235, 1979.
\newblock \href {https://doi.org/10.1016/0020-0190(79)90113-3}
  {\path{doi:10.1016/0020-0190(79)90113-3}}.

\bibitem{blaser2020complexity}
Markus Bl{\"{a}}ser, Julian D{\"{o}}rfler, and Christian Ikenmeyer.
\newblock On the complexity of evaluating highest weight vectors.
\newblock In {\em 36th Computational Complexity Conference ({CCC} 2021)},
  volume 200 of {\em LIPIcs}, pages 29:1--29:36. Dagstuhl, 2021.
\newblock \href {https://doi.org/10.4230/LIPIcs.CCC.2021.29}
  {\path{doi:10.4230/LIPIcs.CCC.2021.29}}.

\bibitem{BlekTeit:MaximumTypicalGenericRanks}
Grigoriy Blekherman and Zach Teitler.
\newblock On maximum, typical and generic ranks.
\newblock {\em Math. Ann.}, 362(3-4):1021--1031, 2015.
\newblock \href {https://doi.org/10.1007/s00208-014-1150-3}
  {\path{doi:10.1007/s00208-014-1150-3}}.

\bibitem{BIZ18}
Karl Bringmann, Christian Ikenmeyer, and Jeroen Zuiddam.
\newblock On algebraic branching programs of small width.
\newblock {\em J. {ACM}}, 65(5):32:1--32:29, 2018.
\newblock \href {https://doi.org/10.1145/3209663} {\path{doi:10.1145/3209663}}.

\bibitem{BuczBucz:SecantVarsHighDegVeroneseReembeddingsCataMatAndGorSchemes}
Weronika Buczy\'{n}ska and Jaros\l{}aw Buczy\'{n}ski.
\newblock Secant varieties to high degree {V}eronese reembeddings,
  catalecticant matrices and smoothable {G}orenstein schemes.
\newblock {\em J. Algebraic Geom.}, 23(1):63--90, 2014.
\newblock \href {https://doi.org/10.1090/S1056-3911-2013-00595-0}
  {\path{doi:10.1090/S1056-3911-2013-00595-0}}.

\bibitem{BB-wild}
Weronika Buczy\'{n}ska and Jaros\l{}aw Buczy\'{n}ski.
\newblock On differences between the border rank and the smoothable rank of a
  polynomial.
\newblock {\em Glasg. Math. J.}, 57(2):401--413, 2015.
\newblock \href {https://doi.org/10.1017/S0017089514000378}
  {\path{doi:10.1017/S0017089514000378}}.

\bibitem{MR3724212}
Jaros\l~aw Buczy\'{n}ski and Joachim Jelisiejew.
\newblock Finite schemes and secant varieties over arbitrary characteristic.
\newblock {\em Differential Geom. Appl.}, 55:13--67, 2017.
\newblock \href {https://doi.org/10.1016/j.difgeo.2017.08.004}
  {\path{doi:10.1016/j.difgeo.2017.08.004}}.

\bibitem{BB-apolarity}
Weronika Buczyńska and Jarosław Buczyński.
\newblock Apolarity, border rank, and multigraded {Hilbert} scheme.
\newblock {\em Duke Math. J.}, 170(16):3659--3702, 2021.
\newblock \href {https://doi.org/10.1215/00127094-2021-0048}
  {\path{doi:10.1215/00127094-2021-0048}}.

\bibitem{BBT-monomials}
Weronika Buczyńska, Jarosław Buczyński, and Zach Teitler.
\newblock Waring decompositions of monomials.
\newblock {\em Journal of Algebra}, 378:45--57, 2013.
\newblock \href {https://doi.org/10.1016/j.jalgebra.2012.12.011}
  {\path{doi:10.1016/j.jalgebra.2012.12.011}}.

\bibitem{Bur04}
Peter B{\"{u}}rgisser.
\newblock The complexity of factors of multivariate polynomials.
\newblock {\em Found. Comput. Math.}, 4(4):369--396, 2004.
\newblock \href {https://doi.org/10.1007/s10208-002-0059-5}
  {\path{doi:10.1007/s10208-002-0059-5}}.

\bibitem{bcs97}
Peter B{\"u}rgisser, Michael Clausen, and Mohammad~Amin Shokrollahi.
\newblock {\em Algebraic Complexity Theory}, volume 315 of {\em Grundlehren der
  mathematischen Wissenschaften}.
\newblock Springer-Verlag, 1997.
\newblock \href {https://doi.org/10.1007/978-3-662-03338-8}
  {\path{doi:10.1007/978-3-662-03338-8}}.

\bibitem{carlini2006reducing}
Enrico Carlini.
\newblock Reducing the number of variables of a polynomial.
\newblock In {\em Algebraic geometry and geometric modeling}, pages 237--247.
  Springer, Berlin, 2006.
\newblock \href {https://doi.org/10.1007/978-3-540-33275-6_15}
  {\path{doi:10.1007/978-3-540-33275-6_15}}.

\bibitem{Cay:TheoryLinTransformations}
Arthur Cayley.
\newblock {On the theory of linear transformations}.
\newblock {\em Cambridge Math. J.}, IV:193--209, 1845.
\newblock Reprinted in A. Cayley, \emph{The Collected Mathematical Papers I},
  80--94, Cambridge University Press, 1889.
\newblock \href {https://doi.org/10.1017/CBO9780511703676.014}
  {\path{doi:10.1017/CBO9780511703676.014}}.

\bibitem{CHIL18}
Luca Chiantini, Jonathan~D. Hauenstein, Christian Ikenmeyer, Joseph~M.
  Landsberg, and Giorgio Ottaviani.
\newblock Polynomials and the exponent of matrix multiplication.
\newblock {\em Bull. LMS}, 50(3):369--389, 2018.
\newblock \href {https://doi.org/10.1112/blms.12147}
  {\path{doi:10.1112/blms.12147}}.

\bibitem{Cleb:TheorieFlachen}
Alfred Clebsch.
\newblock {Zur Theorie der algebraischen Flächen}.
\newblock {\em J. Reine Angew. Math.}, 58:93--108, 1861.
\newblock \href {https://doi.org/10.1515/crll.1861.58.93}
  {\path{doi:10.1515/crll.1861.58.93}}.

\bibitem{duttademystifying}
Pranjal Dutta, Prateek Dwivedi, and Nitin Saxena.
\newblock Demystifying the border of depth-3 algebraic circuits.
\newblock In {\em {FOCS} 2021}, pages 92--103. {IEEE}, 2021.
\newblock \href {https://doi.org/10.1109/FOCS52979.2021.00018}
  {\path{doi:10.1109/FOCS52979.2021.00018}}.

\bibitem{EGOW18}
Klim Efremenko, Ankit Garg, Rafael~Mendes de~Oliveira, and Avi Wigderson.
\newblock Barriers for rank methods in arithmetic complexity.
\newblock In {\em 9th Innovations in Theoretical Computer Science Conference
  ({ITCS} 2018)}, volume~94 of {\em LIPIcs}, pages 1:1--1:19. Dagstuhl, 2018.
\newblock \href {https://doi.org/10.4230/LIPIcs.ITCS.2018.1}
  {\path{doi:10.4230/LIPIcs.ITCS.2018.1}}.

\bibitem{Forbes16}
Michael Forbes.
\newblock Some concrete questions on the border complexity of polynomials.
\newblock Presentation given at the Workshop on Algebraic Complexity Theory
  ({WACT} 2016) in {T}el {A}viv, 2016.
\newblock URL: \url{https://www.youtube.com/watch?v=1HMogQIHT6Q}.

\bibitem{Galazka:VectorBundlesGiveEqnsForCactus}
Maciej Ga{\l}\k{a}zka.
\newblock {Vector bundles give equations of cactus varieties}.
\newblock {\em Lin. Alg. Appl.}, 521:254–262, 2017.
\newblock \href {https://doi.org/10.1016/j.laa.2016.12.005}
  {\path{doi:10.1016/j.laa.2016.12.005}}.

\bibitem{DBLP:conf/focs/GargMOW19}
Ankit Garg, Visu Makam, Rafael~Mendes de~Oliveira, and Avi Wigderson.
\newblock More barriers for rank methods, via a "numeric to symbolic" transfer.
\newblock In {\em {FOCS} 2019}, pages 824--844. {IEEE}, 2019.
\newblock \href {https://doi.org/10.1109/FOCS.2019.00054}
  {\path{doi:10.1109/FOCS.2019.00054}}.

\bibitem{grace_young_2010}
John~Hilton Grace and Alfred Young.
\newblock {\em The Algebra of Invariants}.
\newblock Cambridge Library Collection --- Mathematics. Cambridge University
  Press, 2010.
\newblock \href {https://doi.org/10.1017/CBO9780511708534}
  {\path{doi:10.1017/CBO9780511708534}}.

\bibitem{grothendieck1962techniques}
Alexander Grothendieck.
\newblock Techniques de construction et th\'eor\`emes d'existence en
  g\'eom\'etrie alg\'ebrique {IV} : les sch\'emas de {Hilbert}.
\newblock In {\em S\'eminaire Bourbaki : ann\'ees 1960/61, expos\'es 205-222},
  number~6 in S\'eminaire Bourbaki. Soci\'et\'e math\'ematique de France, 1961.
\newblock URL: \url{http://www.numdam.org/item/SB_1960-1961__6__249_0/}.

\bibitem{Gurvits-permanent}
Leonid Gurvits.
\newblock Ryser (or polarization) formula for the permanent is essentially
  optimal: the {Waring} rank approach.
\newblock Technical Report LA-UR08-06583, Los Alamos National Laboratory, 2008.

\bibitem{haiman2004multigraded}
Mark Haiman and Bernd Sturmfels.
\newblock Multigraded {Hilbert} schemes.
\newblock {\em Journal of Algebraic Geometry}, 13(4):725--769, 2004.
\newblock \href {https://doi.org/10.1090/S1056-3911-04-00373-X}
  {\path{doi:10.1090/S1056-3911-04-00373-X}}.

\bibitem{Hartshorne}
Robin Hartshorne.
\newblock {\em Algebraic geometry}.
\newblock Springer, New York, 1977.
\newblock Graduate Texts in Mathematics, No. 52.
\newblock \href {https://doi.org/10.1007/978-1-4757-3849-0}
  {\path{doi:10.1007/978-1-4757-3849-0}}.

\bibitem{Iarrobino_1995}
Anthony Iarrobino.
\newblock Inverse system of a symbolic power {II}. the {W}aring problem for
  forms.
\newblock {\em Journal of Algebra}, 174(3):1091--1110, 1995.
\newblock \href {https://doi.org/10.1006/jabr.1995.1169}
  {\path{doi:10.1006/jabr.1995.1169}}.

\bibitem{Iarrobino-Kanev}
Anthony Iarrobino and Vassil Kanev.
\newblock {\em Power sums, {G}orenstein algebras, and determinantal loci},
  volume 1721 of {\em Lecture Notes in Mathematics}.
\newblock Springer-Verlag, Berlin, 1999.
\newblock \href {https://doi.org/10.1007/BFb0093426}
  {\path{doi:10.1007/BFb0093426}}.

\bibitem{IS22}
Christian Ikenmeyer and Abhiroop Sanyal.
\newblock A note on {VNP}-completeness and border complexity.
\newblock {\em Inf. Process. Lett.}, 176:106243, 2022.
\newblock \href {https://doi.org/10.1016/j.ipl.2021.106243}
  {\path{doi:10.1016/j.ipl.2021.106243}}.

\bibitem{Jelisiejew-maxrank}
Joachim Jelisiejew.
\newblock An upper bound for the {W}aring rank of a form.
\newblock {\em Arch. Math. (Basel)}, 102(4):329--336, 2014.
\newblock \href {https://doi.org/10.1007/s00013-014-0632-6}
  {\path{doi:10.1007/s00013-014-0632-6}}.

\bibitem{Jeli:Pathologies}
Joachim Jelisiejew.
\newblock {Pathologies on the Hilbert scheme of points}.
\newblock {\em Inventiones mathematicae}, 220(2):581–610, 2020.
\newblock \href {https://doi.org/10.1007/s00222-019-00939-5}
  {\path{doi:10.1007/s00222-019-00939-5}}.

\bibitem{JelMan:LimitsSaturatedIdeals}
Joachim Jelisiejew and Tomasz Mańdziuk.
\newblock Limits of saturated ideals, 2022.
\newblock \href {http://arxiv.org/abs/2210.13579} {\path{arXiv:2210.13579}}.

\bibitem{kayal2007polynomial}
Neeraj Kayal and Nitin Saxena.
\newblock Polynomial identity testing for depth 3 circuits.
\newblock {\em comput. complex.}, 16(2):115--138, 2007.
\newblock \href {https://doi.org/10.1007/s00037-007-0226-9}
  {\path{doi:10.1007/s00037-007-0226-9}}.

\bibitem{Kra85}
Hanspeter Kraft.
\newblock {\em Geometrische {M}ethoden in der {I}nvariantentheorie}.
\newblock Springer, Berlin, 2 edition, 1985.
\newblock \href {https://doi.org/10.1007/978-3-663-10143-7}
  {\path{doi:10.1007/978-3-663-10143-7}}.

\bibitem{kum20}
Mrinal Kumar.
\newblock On the power of border of depth-3 arithmetic circuits.
\newblock {\em {ACM} Trans. Comput. Theory}, 12(1):5:1--5:8, 2020.
\newblock \href {https://doi.org/10.1145/3371506} {\path{doi:10.1145/3371506}}.

\bibitem{landsberg2010}
Joseph~M. Landsberg and Zach Teitler.
\newblock On the ranks and border ranks of symmetric tensors.
\newblock {\em Found. Comput. Math.}, 10(3):339--366, 2010.
\newblock \href {https://doi.org/10.1007/s10208-009-9055-3}
  {\path{doi:10.1007/s10208-009-9055-3}}.

\bibitem{lehmkuhl1989order}
Thomas Lehmkuhl and Thomas Lickteig.
\newblock On the order of approximation in approximative triadic decompositions
  of tensors.
\newblock {\em Theor. Comput. Sci.}, 66(1):1--14, 1989.
\newblock \href {https://doi.org/10.1016/0304-3975(89)90141-2}
  {\path{doi:10.1016/0304-3975(89)90141-2}}.

\bibitem{MS01}
Ketan Mulmuley and Milind~A. Sohoni.
\newblock Geometric complexity theory {I:} an approach to the {P} vs. {NP} and
  related problems.
\newblock {\em {SIAM} J. Comput.}, 31(2):496--526, 2001.
\newblock \href {https://doi.org/10.1137/S009753970038715X}
  {\path{doi:10.1137/S009753970038715X}}.

\bibitem{nis91}
Noam Nisan.
\newblock Lower bounds for non-commutative computation (extended abstract).
\newblock In {\em STOC'91: Proceedings of the 23rd Annual {ACM} Symposium on
  Theory of Computing}, pages 410--418. {ACM}, 1991.
\newblock \href {https://doi.org/10.1145/103418.103462}
  {\path{doi:10.1145/103418.103462}}.

\bibitem{Oldenburger}
Rufus Oldenburger.
\newblock Polynomials in several variables.
\newblock {\em Annals of Mathematics}, 41(3):694--710, 1940.
\newblock \href {https://doi.org/10.2307/1968741} {\path{doi:10.2307/1968741}}.

\bibitem{Palatini:SuperficieAlg}
Francesco Palatini.
\newblock {Sulle superficie algebriche i cui $S_h$ $(h+1)$-seganti non
  riempiono lo spazio ambiente}.
\newblock {\em Atti della R. Acc. delle Scienze di Torino}, 41:634–640, 1906.
\newblock URL:
  \url{https://archive.org/details/attidellarealeac41real/page/634}.

\bibitem{sax08}
Nitin Saxena.
\newblock Diagonal circuit identity testing and lower bounds.
\newblock In {\em Automata, Languages and Programming, 35th International
  Colloquium, {ICALP} 2008, Part {I}}, volume 5125 of {\em Lecture Notes in
  Computer Science}, pages 60--71. Springer, 2008.
\newblock \href {https://doi.org/10.1007/978-3-540-70575-8_6}
  {\path{doi:10.1007/978-3-540-70575-8_6}}.

\bibitem{Sylv:PrinciplesCalculusForms}
James~Joseph Sylvester.
\newblock {On the principles of the calculus of forms}.
\newblock {\em J. Cambridge and Dublin Math.}, 7:52--97, 1852.
\newblock Reprinted in J.\,J.\,Sylvester, \emph{The Collected Mathematical
  Papers I}, 284--327, Cambridge University Press, 1904.

\bibitem{Terr:seganti}
Alessandro Terracini.
\newblock {Sulle $V_k$ per cui la varietà degli $S_h$ $(h+1)$-seganti ha
  dimensione minore dell'ordinario}.
\newblock {\em Rend. Circ. Mat.}, 31:392–396, 1911.
\newblock \href {https://doi.org/10.1007/BF03018812}
  {\path{doi:10.1007/BF03018812}}.

\bibitem{wact-open-problems}
{WACT 2016}.
\newblock Some accessible open problems.
\newblock URL:
  \url{https://www.cs.tau.ac.il/~shpilka/wact2016/concreteOpenProblems/openprobs.pdf}.

\end{thebibliography}

\appendix

\section{Behind the scenes: generalized additive decompositions and schemes}\label{sec:apolarity}

We will now discuss how the results similar to these of the~\cref{sec:deborderingWR} can be obtained from apolarity theory and the study of $0$-dimensional schemes in projective space.
The connection between variations of Waring rank, apolar schemes and generalized additive decompositions is explored in detail by Bernardi, Brachat and Mourrain in~\cite{BBM-ranks} (they use a subtly different notion of generalized affine decomposition).
In particular, there exists a much stronger version of~\Cref{lem:borderrank-gad} (although with a slightly worse restriction on degree), which tightly relates generalized additive decompositions of a homogeneous polynomial $f$ to its \emph{cactus rank} $\CR(f)$, a variation of Waring rank arising in apolarity theory defined in terms of $0$-dimensional schemes in the space of linear forms in place of finite sets of linear forms.
We will formally define the notions of cactus rank and size of a generalized additive decomposition later, for not let us state the theorem, which is based on~\cite[Thm.~3.5]{BBM-ranks}.

\begin{theorem}\label{thm:cactusrank}
If $\deg f \geq 2 \cdot \CR(f) - 1$, then the cactus rank of a homogeneous polynomial $f$ is equal to the minimal possible size of a generalized additive decomposition for $f$.
\end{theorem}

To connect cactus rank to border rank we need and intermediate notion of \emph{smoothable rank} $\SR(f)$.
Smoothable rank is an upper bound on cactus rank, and it coincides with border rank for polynomials of high enough degree.
\begin{theorem}[\cite{BB-wild}]\label{thm:smoothrank}
If $\deg f \geq \bwr(f) - 1$, then $\bwr(f) = \SR(f)$.
\end{theorem}

The goal of this section is to review the basic notions of apolarity theory, define cactus rank, smoothable rank and size of a generalized additive decomposition, and explain the ideas behind the proof of~\cref{thm:cactusrank} stated above.

\subsection{Some notation}
Let us fix the notation.
Let $S = \bbC[x_1, \dots, x_n]$ be the algebra of polynomials and $T = \bbC[\partial_1, \dots, \partial_n]$ be the algebra of polynomial differential operators with constant coefficients (referred to as \emph{diffoperators} in what follows), which acts on $S$ in the standard way.

Denote by $V$ the space of linear forms $S_1$.
We identify $T_1$ with the dual space $V^*$.
More generally, the action of $T$ on $S$ gives rise to a nondegenerate pairing between the homogeneous parts $S_d$ and $T_d$ for every $d$.
We use orthogonality with respect to this pairing,
that is, for a subset $F \subset S_d$ we denote $F^{\perp} = \{ \alpha \in T_d \mid \alpha \cdot f = 0 \text{ for all $f \in F$}\}$, and vice versa, for a subset $D \subset T_d$ we let $D^{\perp} = \{ f \in S_d \mid \alpha \cdot f = 0 \text{ for all $\alpha \in D$}\}$

\subsection{Projective geometry}

The algebra $T$ is isomorphic to $\bbC[V]$, the algebra of polynomials in the coefficients of linear forms.
The isomorphism maps a homogeneous element $\alpha \in T_d$ to $\bar{\alpha} \in \bbC[V]_d$ defined as $\bar{\alpha}(\ell) = \alpha \cdot \frac{\ell^d}{d!}$.

Recall that a homogeneous ideals in $T \cong \bbC[V]$ are in correspondence with subsets of the projective space $\bbP V$.
More specifically, projective varieties are subsets of $\bbP V$ defined by vanishing of some set of polynomials.
The set of all polynomials vanishing on a projective variety $Z$ is a homogeneous ideal $I$, which is saturated ($\alpha T_1 \subset I \Rightarrow \alpha \in I$) and radical ($\alpha^n \in I \Rightarrow \alpha \in I$).
If we consider an ideal $I$ which is saturated but not radical, it defines a projective scheme, which coincides with the variety defined by $I$ as a topological space, but has additional structure which distinguishes it from this variety.

If $I \subset T$ is a homogeneous ideal, then the function $h_I(p) = \dim (T_p / I_p)$ is called the \emph{Hilbert function} of $I$.
The Hilbert function of a homogeneous ideal $I$ always coincides with some polynomial $H_I(p)$ for $p$ large enough.
This polynomial is called the \emph{Hilbert polynomial} of $I$.
Many topological and geometric properties of a projective variety or a scheme can be deduced from its Hilbert polynomial, in particular, its dimension and degree~\cite[\S I.7]{Hartshorne}.
We are specifically interested in ideals with constant Hilbert polynomials.
These ideals corresponds to schemes of dimension $0$.
This means that a variety with Hilbert polynomial~$r$ is a set of $r$ distinct points in $\bbP V$.
In algebra, ideals with constant Hilbert polynomial $H_I = r$ are referred to as \emph{ideals of Krull dimension 1} (the mismatch with the dimension of a scheme is because in algebra dimension is counted in affine space).
The number $r$ is referred to as the \emph{length} of the ideal or the \emph{degree} of the corresponding $0$-dimensional scheme.

\subsection{Apolarity theory}

The connection between Waring rank and algebraic geometry is provided by the apolarity theory, which has its source in the works of Sylvester and Macaulay.

\begin{definition}
The \emph{apolar ideal} of a polynomial $f \in S$ is an ideal in $T$ defined as $\Ann(f) = \{ \alpha \in T \mid \alpha \cdot f = 0 \}$.
The \emph{apolar algebra} of $f$ is $A(f) = T / \Ann(f)$.
An ideal $I \subset T$ is said to be \emph{apolar} to $f$ if it lies in $\Ann(f)$.
A scheme $Z \subset \bbP V$ is \emph{apolar} to $f$ if its defining ideal is.
\end{definition}
Note that as a vector space, $A(f)$ is isomorphic to the space of all partial derivatives $T \cdot f$ via $(\alpha + \Ann(f)) \mapsto \alpha \cdot f$.

To relate apolarity to Waring rank, we also define an ideal associated with a set of linear forms.
Given $r$ linear forms $\ell_1, \dots, \ell_r$, consider the sequences of subspaces $E_p = \Span(\{\ell_1^{p}, \dots, \ell_r^{p}\}) \subset S_p$ and $I_p = E_p^{\perp} \subset T_p$.
An important fact is that $I = \bigoplus_{p = 0}^{\infty} I_p$ is a homogeneous ideal in $T$.
From the geometric point of view it can be described as the vanishing ideal of the set $Z = \{ [\ell_1], \dots, [\ell_r] \}$ in the projective space $\bbP V$.
Algebraically, the fact that $I$ is a homogeneous ideal follows from the following useful proposition.

\begin{proposition}
A sequence of subspaces $E_p \subset S_p$ satisfies the property $T_1 \cdot E_{p + 1} \subset E_p$ if and only if $I = \bigoplus_{p = 0}^{\infty} E_p^{\perp}$ is a homogeneous ideal.
If this is the case, then $h_I(p) = \dim E_p$.
\end{proposition}
\begin{proof}
Let $I_p = E_p^{\perp}$.
The fact that $I$ is a homogeneous ideal can be written as $I_{p + 1} \supset T_1 \cdot I_p$, which is equivalent to $T_1 \cdot E_{p + 1} \subset E_p$, as both of these statements reduce to
\[
(\alpha \partial) \cdot f = \alpha \cdot (\partial f) = 0 \text{ for all $\alpha \in I_p, \partial \in T_1, f \in E_{p + 1}$.}
\]
For the Hilbert function, note that $\dim (T_p / I_p) = \dim T_p - \dim I_p = \dim I_p^{\perp} = \dim E_p$.
\end{proof}

\begin{lemma}[Apolarity lemma]
$f \in S_d$ is apolar to a homogeneous ideal $I = \bigoplus_{p = 0}^{\infty} E_p^{\perp}$ if and only if $f \in E_d$.
\end{lemma}
\begin{proof}
If $I$ is apolar to $f$, then $I_d \subset \Ann(f)_d$ and therefore $E_d \supset (\Ann(f)_d)^{\perp} = f^{\perp\perp}\ni f$.

For the other direction, let $f \in E_d$.
Note that $\Ann(f)_p = T_p$ for $p > d$, so we only need to check $I_p \subset \Ann(f)$ for $p \leq d$.
Note that if we have $\alpha \cdot f \neq 0$  for $\alpha \in T_p$ with $p < d$, then there exists $\partial \in T_1$ such that $\partial \alpha \cdot f = \partial \cdot (\alpha \cdot f) \neq 0$.
This can be restated as $T_1 \alpha \subset \Ann(f) \Rightarrow \alpha \in \Ann(f)$ for all $\alpha \in T_p$ with $p < d$.

For $p \leq d$ we have $\alpha \in I_p \Rightarrow T_1^{d - p} \alpha \subset I_d = E_d^{\perp} \Rightarrow T_1^{d - p} \alpha \cdot f = 0 \Rightarrow \alpha \in \Ann(f)$, which proves $I_p \subset \Ann(f)$.
\end{proof}

\begin{corollary}[Apolarity for Waring rank]
$f \in S_d$ has $\WR(f) \leq r$ if and only if $f$ is apolar to the vanishing ideal of $r$ points in~$\bbP V$.
\end{corollary}
\begin{proof}
By definition, $\WR(f) \leq r$ if and only if there exists linear forms $\ell_1, \dots, \ell_r$ such that $f \in \Span(\{\ell_1^d, \dots, \ell_r^d\})$.
As was discussed before, the vanishing ideal $I$ of the set $Z = \{[\ell_1], \dots, [\ell_r]\}$ has $E_d = I_d^{\perp} = \Span(\{\ell_1^d, \dots, \ell_r^d\})$, and by apolarity lemma $f \in E_d$ is equivalent to $f$ being apolar to $I$.
\end{proof}

\subsection{Families of subspaces, ideals and their limits}

Before considering border Waring rank, we need to define limits of families of subspaces and families of ideals.

Let $W$ be a vector space. We consider two types of families of subspaces in $W$.
First is a family of subspaces of the form $E(\eps) = \Span(\{w_1(\eps), \dots, w_r(\eps)\})$ where $w_k(\eps)$ are families of vectors in $W$ with coordinates given by rational functions of $\eps$.
We write $w_k \in W(\eps)$ in this case.
The second type is a family $E(\eps) = \{ w \mid y_1(\eps; w) = \dots = y_q(\eps; w) = 0 \}$ of vector spaces defined by linear forms $y_1, \dots, y_q \in W^*(\eps)$ which again depend rationally on the parameter $\eps$.

In both cases we define the limit $\hat{E} = \lim_{\eps \to 0} E(\eps)$ as the subspace containing the limits of all families $w \in W(\eps)$ such that $w(\eps) \in E(\eps)$ for $\eps \neq 0$ (whenever $E(\eps)$ and $w(\eps)$ are defined).

For $E(\eps) = \Span(\{w_1(\eps), \dots, w_r(\eps) \})$ from semicontinuity of rank we have that the maximal possible value of $\dim E(\eps)$ is attained on an open set of values of $\eps$.
The situation is opposite for the family of the second type $E(\eps) = \Span(\{y_1(\eps), \dots, y_q(\eps) \})^{\perp}$.
In both cases the dimension of $\hat{E}$ cannot be higher then the generic dimension.
Indeed, if $\hat{E}$ contains linearly indeoendent vectors $v_1, \dots, v_m$, then there are families $v_1(\eps), \dots, v_m(\eps)$ which have them as limits, and these families will be linearly independent for an open subset of values of $\eps$.
Considering two families $E(\eps) \subset W$ and $E(\eps)^{\perp} \subset W^*$ together, we see that $\dim \hat{E}$ is actually equal to the generic dimension of $E(\eps)$ (maximal dimension for the families of the first type, and minimal --- for the families of the second type).

Alternatively, we may associate with a family of subspaces a family of points in the \emph{Grassmannian} -- the space of all $k$-dimensional subspaces in $W$.
The Grassmannian can be defined as the projective variety in $\bbP \Lambda^k W$ consisting of all points of the form $[w_1 \wedge \dots \wedge w_k]$, which represent $k$-dimensional subspaces spanned by $w_1, \dots, w_k$ respectively.
If $E(\eps)$ is a family with generic dimension $k$ and $v_1(\eps), \dots, v_k(\eps) \in E(\eps)$ are linearly independent for generic values of $\eps$, then we can define a rational map $\eps \mapsto [v_1(\eps) \wedge \dots \wedge v_k(\eps)]$ and take the limit of this map in the Grassmannian.

Suppose $I(\eps)$ is a family of homogeneous ideals in $T$, that is, $I(\eps) = \bigoplus_{p = 0}^{\infty} I_p(\eps)$ for the families of subspaces $I_p(\eps) \subset T_p$ such that $I_{p + 1}(\eps) \supset I_p(\eps) \cdot T_1$.
By continuity of multiplication for the limit subspaces $\hat{I}_p = \lim_{\eps \to 0} I_p(\eps)$ we still have $\hat{I}_{p + 1} \supset \hat{I}_p \cdot T_1$.
Hence $\hat{I}$ is again a homogeneous ideal in $T$.
This notion of limit of ideals corresponds to taking limits in the \emph{multigraded Hilbert scheme}, which is a space of ideals with given Hilbert function, see~\cite{haiman2004multigraded}.
We refer to this limit as the \emph{multigraded limit of a family of ideals}.
The problem is that the limit in the multigraded Hilbert scheme can be non-saturated and thus not correspond to a geometric object in projective space.

For example, consider three families of points $(1:0:1), (-1:0:1), (0:\eps:1)$ in $\bbP^2$.
The family of vanishing ideals is $\left<x_1 x_2, x_2 (x_2 - \eps x_3), \eps(x_1^2 - x_3^2) + x_2 x_3, x_1^3 - x_1x_3^2 \right>$.
Taking $\eps \to 0$ we obtain the ideal $\left<x_1 x_2, x_2^2, x_2 x_3, x_1^3 - x_1x_3^2 \right>$, which is not saturated, since it contains $x_1 x_2, x_2^2, x_2 x_3$ but not $x_2$.
Taking the saturation, we obtain $\left<x_2, x_1^3 - x_1x_3^2\right>$ which corresponds to three points $(1:0:1), (-1:0:1), (0:0:1)$ as expected.

We can take saturation after obtaining the limit ideal.
This notion of limit corresponds to limits in the \emph{Hilbert scheme}, which is the space of ideals with the fixed Hilbert polynomial.
It was defined by Grothendieck~\cite{grothendieck1962techniques}, see also~\cite[Appx.C]{Iarrobino-Kanev}.

\subsection{Border apolarity}

We will now describe the basic idea of the apolarity theory for border Waring rank, which was developed by Buczy\'{n}ska and Buczy\'{n}ski in~\cite{BB-apolarity}.

Let $f = \lim_{\eps \to 0} \sum_{k = 1}^{r} \ell_k^d$ be a border Waring rank decomposition. Consider the families of subspaces $E_p(\eps) = \Span(\{\ell_1(\eps)^p, \dots, \ell_r(\eps)^p \}) \subset S_p$ and the family of homogeneous ideals $I(\eps) = \bigoplus_{p = 0}^{\infty} E_p(\eps)^{\perp}$ in $T$.

As $\eps \to 0$, we obtain a sequence of subspaces $\hat{E}_p = \lim_{\eps \to 0} E_p(\eps) \subset S_p$ and a homogeneous ideal $\hat{I} = \lim_{\eps \to 0} I(\eps)$ (taking the limit in the multigraded Hilbert scheme).
Let $\bar{f} = \sum_{k = 1}^{r} \ell_k^d \in S_d(\eps)$, so that $f = \lim_{\eps \to 0} \bar{f}(\eps)$.
By the Apolarity Lemma the ideal $I(\eps)$ is apolar to $\bar{f}(\eps)$ for $\eps \neq 0$, which means that $\alpha(\eps) \cdot \bar{f}(\eps) = 0$ for every $\alpha(\eps) \in I(\eps)$.
Since the action of $T$ on $S$ is continuous, we obtain from this $(\lim_{\eps \to 0} \alpha(\eps)) \cdot f = 0$, if the limit exists.
Thus $\hat{I}$ is apolar to $f$.

On the other hand, suppose that $f \in S_d$ is apolar to an ideal $\hat{I}$ which is a \emph{limit of ideals of $r$ points}, that is, there exists a family $I(\eps)$ such that $I(\eps)$ is the vanishing ideal of a set of $r$ points in $\bbP V$.
Define $E_d(\eps) = I(\eps)_d^{\perp} \subset S_d$.
For $\eps \neq 0$ the subspace $E_d(\eps)$ is a span of powers of $r$ linear forms, so it consists of polynomials with Waring rank at most $r$.
Since $f$ is orthogonal to $\hat{I}_d$, it lies in the limit $\lim_{\eps \to 0} E_d(\eps)$ and thus has border Waring rank at most $r$.

\begin{theorem}[{Border apolarity, \cite{BB-apolarity}}]
$f \in S_d$ has $\bwr(f) \leq r$ if and only if $f$ is apolar to an ideal $\hat{I}$ which is a limit of ideals of $r$ points. 
\end{theorem}

\subsection{Various ranks via apolarity.}

The apolarity lemma provides a template for defining different notions of rank for homogeneous polynomials by varying the class of ideals apolar to $f$.

\begin{definition}
Let $\mathcal{C}$ be a class of ideals of Krull dimension $1$.
If $f \in S_d$ is a homogeneous polynomial, we define the $\mathcal{C}$-rank of $f$ as the minimal $r$ such that there exists an ideal $I \subset \mathcal{C}$ apolar to $f$ with length $r$.
\end{definition}

As we have seen, Waring rank and border Waring rank are special cases of this definition corresponding to ideals of points and their limits.

We are now ready to define cactus rank and smoothable rank.
The cactus rank $\CR(f)$ is obtained from the template definition above if we consider the class of all saturated ideals with constant Hilbert polynomial, that is, ideals of $0$-dimensional schemes.
The smoothable rank $\SR(f)$ corresponds to saturated limits of ideals of points.
In addition, the border cactus rank $\bCR(f)$ is defined by considering limits of saturated ideals.

\medskip

\begin{tabular}{|l|l|l|}
\hline
Class of ideals  & Rank & Notation \\
\hline
Ideals of points (radical saturated ideals)  & Waring rank & $\WR(f)$ \\
Limits of ideals of points & Border Waring rank 
& $\bwr(f)$ \\
Smoothable ideals (saturated limits of ideals of points) & Smoothable rank & $\SR(f)$ \\
Saturated ideals & Cactus rank & $\CR(f)$ \\
Saturable ideals (limits of saturated ideals) & Border cactus rank & $\bCR(f)$ \\
\hline
\end{tabular}

\medskip

The unified definition allows us to determine relations between these different ranks.

\begin{theorem}[{\cite{BBM-ranks}}]\label{thm:bbm-ranks}
The following inequalities hold: $\bCR(f) \leq \CR(f) \leq \SR(f) \leq \WR(f)$ and $\bCR(f) \leq \bwr(f) \leq \SR(f) \leq \WR(f)$.
\end{theorem}
\begin{proof}
The inequality $\bwr(f) \leq \SR(f)$ follows from the fact that if the saturation $I^{\sat} \supset I$ is apolar to~$f$, then $I$ is also apolar to $f$.
Other inequalities follow from the containments between corresponding classes of ideals.
\end{proof}

In addition, for high enough degree the border Waring rank coincides with smoothable rank, and border cactus rank coincides with cactus rank. To prove this, we will need the following property of ideals of Krull dimension $1$.

\begin{theorem}[{\cite[Thm.~1.69]{Iarrobino-Kanev}}]\label{thm:dim0-regularity}
If $I$ is a saturated ideal with $H_I = r$, then $h_I$ is nondecreasing and $h_I(p) = r$ for $p \geq r - 1$.
\end{theorem}

Note that this applying this to the ideal of $r$ points $\{[\ell_1], \dots, [\ell_r]\}$ we get the well-known property that $p$-th powers of pairwise non-proportional linear forms are linearly independent for $p \geq r - 1$.

\begin{theorem}[{\cite[Prop. 2.5]{BuczBucz:SecantVarsHighDegVeroneseReembeddingsCataMatAndGorSchemes}}]
If $\deg f \geq \bwr(f) - 1$, then $\bwr(f) = \SR(f)$. Similarly, if $\deg f \geq \bCR(f) - 1$, then $\bCR(f) = \CR(f)$.
\end{theorem}
\begin{proof}
Suppose $\bwr(f) = r$. Let $I$ be the ideal apolar to $f$ which is a limit of ideals of $r$ points, and let $J = I^{\sat}$ be its saturation.
We have $h_I(p) = r$ for $p \geq r - 1$, because this property is true for ideals of points and $I$ is a limit of ideals of points.
In particular, $h_I(d) = r$ for $d = \deg f$.
It is known that $J_p = I_p$ for $p$ large enough, thus $H_{J} = H_I = r$.
We have $r = h_I(d) \leq h_{J}(d) \leq r$, where the first inequality follows from $I \subset I^{\sat} = J$, and the second one from the monotonicity of $h_{J}$.
Therefore $I_d = J_d$, and thus $f \in J_d^{\perp} = I_d^{\perp}$, that is, $f$ is apolar to $J$, the saturated limit of ideals of points.
We obtain $\SR(f) \leq r = \bwr(f)$. Together with the opposite inequality from~\cref{thm:bbm-ranks}, this gives us the required equality.

The proof of the second statement is the same, with ideal of $r$ points replaces by saturated ideals of length $r$.
\end{proof}

\subsection{Size of generalized additive decompositions}
We now describe how to measure the size of a generalized additive decomposition in a way compatible with various notions of rank.

\begin{definition}
The \emph{partial derivative space} of a polynomial $f \in S$ (not necessarily homogeneous) is the vector space $T \cdot f$ spanned by $f$ and all its partial derivatives of all orders.
\end{definition}
\begin{definition}
Let $\ell$ be a linear form and let $\partial \in T_1$ be a partial derivative such that $\partial \ell = 1$.
Denote $T' = \bbC[\ell^{\perp}] \subset T$ and $S' = \bbC[\partial^{\perp}] \subset S$.
We define the \emph{compression} $f_{(\partial, \ell)}$ of a homogeneous polynomial $f \in S_d$ with respect to $\ell$ and $\partial$ as follows.
Write
\[
f = \sum_{i = 0}^d \frac{\ell^i}{i!} f_i.
\]
with $f_i \in S'_{d - i}$.
Then $f_{(\partial, \ell)} = \sum_{i = 0}^d f_i$.
The \emph{compressed size} of $f$ with respect to $(\partial, \ell)$ is the dimension of the partial derivative space $T' \cdot f_{(\partial, \ell)}$
\end{definition}
One can check that the compressed does not depend on the choice of $\partial$ as long as $\partial \ell = 1$; this can be proved by hand, and it is obtained in \cite{Bernardi2018} in a more intrinsic way.

\begin{definition}
The \emph{size} of a generalized additive decomposition
\[
f = \sum_{k = 1}^m \ell_k^{d - r_k + 1} g_k
\]
is defined as the sum of compressed sizes of the summands $\ell_k^{d - r_k + 1} g_k$ with respect to the corresponding linear forms $\ell_k$.
\end{definition}

\subsection{Schemes and generalized additive decompositions}

We will now prove several lemmas which connect generalized additive decompositions and apolar ideals, finishing the proof of~\cref{thm:cactusrank}.

\begin{lemma}[\cite{Bernardi2018}]\label{lem:gad-to-scheme}
Let $\ell$ be a linear form and let $f \in S_d$ be a homogeneous polynomial.
There exists a homogeneous ideal $I$ apolar to $f$ with length equal to the compressed size of $f$ with respect to $\ell$.
\end{lemma}
\begin{proof}
Let $\partial \in T_1$ be such that $\partial \ell = 1$.
Denote $S' = \bbC[\partial^{\perp}]$ and $T' = \bbC[\ell^{\perp}]$ as above.
The rings $S'$ and $T'$ are in the same dual relationship as $S$ and $T$, and $T$ is generated by $T'$ and~$\partial$.

We start from the ideal $J = \Ann(f_{(\partial, \ell)}) \subset T'$ and homogenize it using $\partial$.
That is, define the homogenization map from $T'_{\leq p} = \bigoplus_{j = 0}^p T'_j$ to $T_p$ sending $\alpha = \sum_{j = 0}^p \alpha_j$ to $\sum_{j = 0}^p \partial^{p - j} \alpha_j$.
The homogeneous part $I_p$ of the ideal $I$ is then the image of $J_{\leq p}$ under this homogenization map.

To show that the ideal $I$ is apolar to $f$,
write $f = \sum_{i = 0}^d \frac{\ell^{i}}{i!} f_{d - i}$ with $f_i \in S'_i$.
Then $f_{(\partial, \ell)} = \sum_{i = 0}^d f_i$.
If $\alpha' = \sum_{j = 0}^p \alpha_j \in T'$ and $\alpha = \sum_{j = 0}^p \partial^{p - j} \alpha_j$ is the $\alpha'$ homogenized, then the statement $\alpha' \cdot f_{(\partial, \ell)} = 0$ is equivalent to $\alpha \cdot f = 0$, since they are both equivalent to $\sum_{j = 0}^p \alpha_j f_{j + e} = 0$ for all $e \geq 0$.
Since $J$ is apolar to $f_{(\partial, \ell)}$, $I$ is apolar to $f$.

Since $f_{(\partial, \ell)}$ has degree at most $d$, $J$ contains $T'_{p}$ for $p > d$.
Hence 
\[T' \cdot f_{(\partial, \ell)} \cong A(f_{(\partial, \ell)}) = T'/J \cong T'_{\leq d}/J_{\leq d}
\] as vector spaces, and for $p > d$ we have $T_p / I_p \cong T'_{\leq p}/J_{\leq p} \cong T'_{\leq d}/J_{\leq d}$.
Therefore, for $p$ large enough $\dim T_p / I_p = \dim T' \cdot f_{(\partial, \ell)}$, and $H_I = \dim T' \cdot f_{(\partial, \ell)}$ as required.
\end{proof}

\begin{lemma}\label{lem:scheme-to-gad}
Suppose $f \in S_d$ is apolar to a saturated primary homogeneous ideal $I$ with length $r$.
If $d \geq 2r - 1$, then $f$ has a one-summand generalized additive decomposition of size at most $r$.
\end{lemma}
\begin{proof}
If $I$ is an ideal of Krull dimension $1$, then it defines a $0$-dimensional scheme, and if it is primary, then this scheme is supported at one point $[\ell] \in \mathbb{P}V$.
The ideal corresponding to this point is $J = \Ann(\ell) = \left<\ell^{\perp}\right>$ and we have $J^m \subset I \subset J$ for some $m \leq r$.
For the corresponding dual space $E_p$ with $p \geq m$ we have $E_p \subset (J^m)_p^{\perp} = \{\ell^{p - m} g \mid g \in S_{\leq m} \}$.
Since $f \in E_d$, it has a one-summand generalized additive decomposition $f = \ell^{d - m} g$.

Choose $\partial \in T_1$ such that $\partial \ell = 1$.
Write $f = \sum_{i = 0}^m \frac{\ell^{d - i}}{(d - i)!} f_i$ with $f_i \in \bbC[\partial^{\perp}]_i$.
Then $f_{(\partial, \ell)} = \sum_{i = 0}^m f_i$ has degree at most $m$.
For every $\alpha' = \sum_{j = 0}^m \in \bbC[\ell^{\perp}]_{\leq m}$ and the corresponding homogeneous $\alpha = \sum_{j = 0}^m \partial^{m - j} \alpha_j$ we have 
\[
\alpha' \cdot f_{(\partial, \ell)} = \sum_{j \geq i} \alpha_j \cdot f_i
\]
and 
\[
\alpha \cdot f = \sum_{j \leq i} \frac{\ell^{d - m + j - i}}{(d - m + j - i)!} \alpha_j \cdot f_i.
\]
Therefore, there is an isomorphism between $\partial^* f_{(\partial, \ell)}$ and $T^m \cdot f \subset E_{d - m}$.
Note that $d - m \geq r$.
By \cref{thm:dim0-regularity} we have $r = H_I = \dim E_{d - m} \geq \dim \partial^* f_{(\partial, \ell)}$.
\end{proof}

\begin{proof}[Proof of~\Cref{thm:cactusrank}]
If $\CR(f) \leq r$, then there exists a saturated homogeneous ideal $I$ apolar to $f$ with Hilbert polynomial $r$.
This ideal corresponds to a $0$-dimensional scheme $Z$, which consists of several points.
Each point corresponds to a primary ideal in the primary decomposition $I = I^{(1)} \cap \dots \cap I^{(m)}$, and for the Hilbert polynomials it is true that $H_I = H_{I^{(1)}} + \dots + H_{I^{(m)}}$.
Defining $E_d = I_d^{\perp}$ and $E^{(k)}_d = (I^{(k)}_d)^{\perp}$ we have $E_d = E^{(1)}_d + \dots + E^{(m)}_d$.
Therefore, $f = f^{(1)} + \dots + f^{(m)}$ where $f^{(k)} \in E^{(k)}_d$.
By~\cref{lem:scheme-to-gad} each $f^{(k)}$ contributes one summand to the generalized additive decomposition.
The sizes of these summands are bounded by $H_{I^{(k)}}$, and the total size is bounded by $r$.

Conversely, if $f$ has a generalized additive decomposition of size $r$, then from each summand we can construct an ideal using \cref{lem:gad-to-scheme} and take the intersection of these ideals to get an ideal apolar to $f$ with Hilbert function at most $r$.
\end{proof}

\section{Characterizing small border Waring rank}

The results on generalized additive decompositions from~\S\ref{sec:borderrank} can be used to describe the polynomials of border rank $2$ and $3$, reproving the results of Landsberg and Teitler~\cite[Sec.~10]{landsberg2010}.

\begin{theorem}\label{lem:borderwaring2}
  A polynomial $f$ with $\bwr(f) = 2$ must have the form $\ell_1^d + \ell_2^d$ or $\ell_1^{d - 1} \ell_2$ where $\ell_1$ and $\ell_2$ are linear forms.

  In the first case, every border rank decomposition for $f$ has the form
  \[
    f = (\ell_1 + \eps \hat{\ell}_1)^d + (\ell_2 + \eps \hat{\ell}_2)^d
  \]
  for some $\hat{\ell}_1, \hat{\ell}_2 \in \bbC[[\eps]][\x]_1$.
  
  In the second case, every border rank decomposition for $f$ has the form
  \[
    f = \frac{1}{\eps^M} \left(a \ell_1 + \eps \hat{\ell}_1 + \eps^M (\frac{1}{a^{d - 1}d}\ell_2 + \ell_3)\right)^d - \frac{1}{\eps^M} \left(a \ell_1 + \eps \hat{\ell}_1 + \eps^M (\ell_3 + \eps \hat{\ell}_2)\right)^d 
  \]
  for some $a \in \bbC$, $\ell_3 \in \bbC[\x]_1$ and $\hat{\ell}_1, \hat{\ell}_2 \in \bbC[[\eps]][\x]_1$.
\end{theorem}
\begin{proof}
By~\Cref{lem:borderrank-gad} $f$ has a generalized additive decomposition
\[
f = \sum_{i = 1}^m \ell_i^{d - r_i + 1} g_i
\]
with $\sum_{i = 1}^m r_i = \bwr(f) = 2$, $\deg g_i = r_i - 1$
There are only two possible partitions $\sum r_i = 2$.
In the case $m = 2, r_1 = r_2 = 1$ the generalized additive decomposition is actually a Waring rank decomposition $f = \ell_1^d + \ell_2^d$.
In the case $m = 1, r_1 = 2$ the polynomial $g_1$ is a linear form, renaming it we have $f = \ell_1^{d - 1} \ell_2$.

From the proof of~\Cref{lem:borderrank-gad} it is clear that in the first case the decomposition must be a sum of two local decompositions of rank $1$, and
a local decomposition of rank $1$ is just a power of $\ell + \eps \hat{\ell}$ for some $\hat{\ell} \in \bbC[[\eps]][\x]_1$.

In the second case the decomposition must be local, which means that both summands in the decomposition have the form $\eps^{-M} (a \ell_1 + \eps \hat{\ell})$.
To obtain $\ell_1^{d - 1} \ell_2$ in the limit, the first $M$ terms in each summand must cancel,
and the terms in $\eps^M$ must differ by $\frac{1}{a^{d - 1}d} \ell_2$.
\end{proof}

\begin{theorem}\label{thm:bwr3}
A polynomial with $\bwr(f) = 3$ must have one of the three normal forms: $\ell_1^d + \ell_2^d + \ell_3^d$ or $\ell_1^{d} +  \ell_2^{d - 1}\ell_3^d$ or $\ell_1^{d - 1} \ell_2 + \ell_1^{d - 2} \ell_3^2$ (with $\ell_1^{d - 2} \ell_3^2$ as a special case).
\end{theorem}
\begin{proof}
By~\Cref{lem:borderrank-gad} $f$ has a generalized additive decomposition
\[
f = \sum_{i = 1}^m \ell_i^{d - r_i + 1} g_i
\]
with $\sum_{i = 1}^m r_i = \bwr(f) = 3$, $\deg g_i = r_i - 1$, and $\bwr(\ell_i^{d - r_i + 1} g_i) \leq r_i$.

In the case $m = 3$, $r_1 = r_2 = r_3 = 1$ this is a Waring rank decomposition $f = \ell_1^d + \ell_2^d + \ell_3^d$.

In the case $m = 2$, we can assume $r_1 = 1$, $r_2 = 2$. The generalized additive decomposition becomes $\ell_1^d + \ell_2^{d - 1} \ell_3$, where $\ell_3 = g_2$ is a linear form.

In the case $m = 1$, $r_1 = 3$ we have $f = \ell_1^{d - 2} g_1$ where $g_1$ is a quadratic form, and $\ell_1^{d - 2} g_1$ has at most three-dimensional space of essential variables.
The quadratic form $g_1$ can be presented in one of the following forms: $a \ell_1^2$, $\ell_2^2$, $\ell_1 \ell_2$, $\ell_1 \ell_2 + \ell_3^2$, or $a \ell_1^2 + \ell_2 \ell_3$ for some linear forms $\ell_2, \ell_3$ linearly independent from $\ell_1$.

If $d > 2$, then all cases except the last are covered by the given normal forms, and in the last case the border rank of $\ell_1^{d - 2} g_1$ is at least $4$ if $d > 2$, which can be seen by computing the dimension of the second order partial derivative space, so it cannot appear.
If $d = 2$ then all forms have rank $3$ and are covered by the case $\ell_1^d + \ell_2^d + \ell_3^d$. \end{proof}

\end{document}